\documentclass[12pt]{article}

\usepackage{a4wide}
\usepackage[english]{babel}
\usepackage{graphicx}
\usepackage{amsmath}
\usepackage{amssymb}
\usepackage{wrapfig}
\usepackage{xspace}
\usepackage{enumerate}
\usepackage{array}
\usepackage{longtable}
\usepackage{float}

\graphicspath{{./images/}}

\newtheorem{defin}{Definition}
 
\newtheorem{theo}[defin]{Theorem}
 \newenvironment{theorem}{\begin{theo} \sl}{\end{theo}}
\newtheorem{lem}[defin]{Lemma}
 \newenvironment{lemma}{\begin{lem} \sl}{\end{lem}}
\newtheorem{coro}[defin]{Corollary}
 \newenvironment{corollary}{\begin{coro} \sl}{\end{coro}}
\newtheorem{prop}[defin]{Proposition}

\newenvironment{proof}{\emph{Proof.}}{\hfill $\Box$\\}

\newcommand{\etal}{\emph{et~al.}\xspace}
\newcommand{\keywords}[1]{\noindent\textbf{Keywords}~ #1}

\newcommand{\halfgraph}{half-\ensuremath{\theta_6}-graph\xspace}
\newcommand{\Graph}[1]{\ensuremath{\theta_{(4 k + #1)}}-Graph\xspace}
\newcommand{\graph}[1]{\ensuremath{\theta_{(4 k + #1)}}-graph\xspace}
\newcommand{\canon}[2]{\ensuremath{T_{#1 #2}}}

\newcommand{\const}{\ensuremath{\boldsymbol{c}}\xspace}

\def\Item{\item\abovedisplayskip=0pt\abovedisplayshortskip=0pt~\vspace*{-\baselineskip}}

\title{Towards Tight Bounds on Theta-Graphs}

\author{
Prosenjit Bose\thanks{School of Computer Science, Carleton University. Research supported in part by FQRNT, NSERC, and Carleton University's President's 2010 Doctoral Fellowship. Email: \texttt{jit@scs.carleton.ca, jdecaruf@cg.scs.carleton.ca, morin@scs.carleton.ca, andre@cg.scs.carleton.ca, sander@cg.scs.carleton.ca}.}
\and
\addtocounter{footnote}{-1}
Jean-Lou De Carufel\footnotemark
\and 
\addtocounter{footnote}{-1}
Pat Morin\footnotemark
\and 
\addtocounter{footnote}{-1}
Andr\'e van Renssen\footnotemark
\and
\addtocounter{footnote}{-1}
Sander Verdonschot\footnotemark
}

\date{}

\begin{document}

\maketitle

\begin{abstract}
  We present improved upper and lower bounds on the spanning ratio of $\theta$-graphs with at least six cones. Given a set of points in the plane, a $\theta$-graph partitions the plane around each vertex into $m$ disjoint cones, each having aperture $\theta = 2 \pi/m$, and adds an edge to the `closest' vertex in each cone. We show that for any integer $k \geq 1$, $\theta$-graphs with $4k + 2$ cones have a spanning ratio of $1 + 2 \sin(\theta/2)$ and we provide a matching lower bound, showing that this spanning ratio tight. 

  Next, we show that for any integer $k \geq 1$, $\theta$-graphs with $4k + 4$ cones have spanning ratio at most $1 + 2 \sin(\theta/2) / (\cos(\theta/2) - \sin(\theta/2))$. We also show that $\theta$-graphs with $4k + 3$ and $4k + 5$ cones have spanning ratio at most $\cos (\theta/4) / (\cos (\theta/2) - \sin (3\theta/4))$. This is a significant improvement on all families of $\theta$-graphs for which exact bounds are not known. For example, the spanning ratio of the $\theta$-graph with 7 cones is decreased from at most 7.5625 to at most 3.5132. These spanning proofs also imply improved upper bounds on the competitiveness of the $\theta$-routing algorithm. In particular, we show that the $\theta$-routing algorithm is $(1 + 2 \sin(\theta/2) / (\cos(\theta/2) - \sin(\theta/2)))$-competitive on $\theta$-graphs with $4k + 4$ cones and that this ratio is tight.  

  Finally, we present improved lower bounds on the spanning ratio of these graphs. Using these bounds, we provide a partial order on these families of $\theta$-graphs. In particular, we show that $\theta$-graphs with $4k + 4$ cones have spanning ratio at least $1 + 2 \tan(\theta/2) + 2 \tan^2(\theta/2)$, where $\theta$ is $2 \pi / (4k + 4)$. This is somewhat surprising since, for equal values of $k$, the spanning ratio of $\theta$-graphs with $4k + 4$ cones is greater than that of $\theta$-graphs with $4k + 2$ cones, showing that increasing the number of cones can make the spanning ratio worse. \\

\end{abstract}

\section{Introduction}
A geometric graph $G$ is a graph whose vertices are points in the plane and whose edges are line segments between pairs of points. A graph $G$ is called plane if no two edges intersect properly. Every edge is weighted by the Euclidean distance between its endpoints. The distance between two vertices $u$ and $v$ in $G$, denoted by $\delta_G(u, v)$, or simply $\delta(u, v)$ when $G$ is clear from the context, is defined as the sum of the weights of the edges along the shortest path between $u$ and $v$ in $G$. A subgraph $H$ of $G$ is a $t$-spanner of $G$ (for $t\geq 1$) if for each pair of vertices $u$ and $v$, $\delta_H(u, v) \leq t \cdot \delta_G(u, v)$. The smallest value $t$ for which $H$ is a $t$-spanner is the \emph{spanning ratio} or \emph{stretch factor} of $H$. The graph $G$ is referred to as the {\em underlying graph} of $H$. The spanning properties of various geometric graphs have been studied extensively in the literature (see \cite{BS11,NS-GSN-06} for a comprehensive overview of the topic). 

Given a spanner, however, it is important to be able to route, i.e. find a short path, between any two vertices. A routing algorithm is said to be \emph{$c$-competitive} with respect to $G$ if the length of the path returned by the routing algorithm is not more than $c$ times the length of the shortest path in $G$~\cite{BFRV12}. The smallest value $c$ for which a routing algorithm is $c$-competitive with respect to $G$ is the \emph{routing ratio} of that routing algorithm. 

In this paper, we consider the situation where the underlying graph $G$ is a straightline embedding of the complete graph on a set of $n$ points in the plane with the weight of an edge $(u,v)$ being the Euclidean distance~$|u v|$ between $u$ and $v$. A spanner of such a graph is called a \emph{geometric spanner}. We look at a specific type of geometric spanner: $\theta$-graphs. 

Introduced independently by Clarkson~\cite{Cl87} and Keil~\cite{Keil88}, $\theta$-graphs are constructed as follows (a more precise definition follows in Section~\ref{sec:Preliminaries}): for each vertex $u$, we partition the plane into $m$ disjoint cones with apex $u$, each having aperture $\theta = 2 \pi/m$. When $m$ cones are used, we denote the resulting $\theta$-graph by the $\theta_m$-graph. The $\theta$-graph is constructed by, for each cone with apex $u$, connecting $u$ to the vertex $v$ whose projection onto the bisector of the cone is closest. Ruppert and Seidel~\cite{RS91} showed that the spanning ratio of these graphs is at most $1/(1 - 2 \sin (\theta/2))$, when $\theta < \pi/3$, i.e. there are at least seven cones. This proof also showed that the \emph{$\theta$-routing} algorithm (defined in Section~\ref{sec:Preliminaries}) is $1/(1 - 2 \sin (\theta/2))$-competitive on these graphs. 

Recently, Bonichon~\etal~\cite{BGHI10} showed that the $\theta_6$-graph has spanning ratio 2. This was done by dividing the cones into two sets, positive and negative cones, such that each positive cone is adjacent to two negative cones and vice versa. It was shown that when edges are added only in the positive cones, in which case the graph is called the \halfgraph, the resulting graph is equivalent to the Delaunay triangulation where the empty region is an equilateral triangle. The spanning ratio of this graph is 2, as shown by Chew~\cite{Chew89}. An alternative, inductive proof of the spanning ratio of the \halfgraph was presented by \mbox{Bose~\etal~\cite{BFRV12}}, along with an optimal local competitive routing algorithm on the half-$\theta_6$-graph. 

Tight bounds on spanning ratios are notoriously hard to obtain. The standard Delaunay triangulation (where the empty region is a circle) is a good example. Its spanning ratio has been studied for over 20 years and the upper and lower bounds still do not match. Also, even though it was introduced about 25 years ago, the spanning ratio of the $\theta_6$-graph has only recently been shown to be finite and tight, making it the first and, until now, only $\theta$-graph for which tight bounds are known. 

In this paper, we improve on the existing upper bounds on the spanning ratio of all $\theta$-graphs with at least six cones. First, we generalize the spanning proof of the \halfgraph given by Bose~\etal~\cite{BFRV12} to a large family of $\theta$-graphs: the \graph{2}, where $k \geq 1$ is an integer. We show that the \graph{2} has a tight spanning ratio of $1 + 2 \sin(\theta/2)$ (see Section~\ref{subsec:Theta4k+2}). 

We continue by looking at upper bounds on the spanning ratio of the other three families of $\theta$-graphs: the \graph{3}, the \graph{4}, and the \graph{5}, where $k$ is an integer and at least 1. We show that the \graph{4} has a spanning ratio of at most $1 + 2 \sin(\theta/2) / (\cos(\theta/2) - \sin(\theta/2))$ (see Section~\ref{subsec:Theta4k+4}). We also show that the \graph{3} and the \graph{5} have spanning ratio at most $\cos (\theta/4) / (\cos (\theta/2) - \sin (3\theta/4))$ (see Section~\ref{subsec:Theta4k+35}). As was the case for Ruppert and Seidel, the structure of these spanning proofs implies that the upper bounds also apply to the competitiveness of $\theta$-routing on these graphs. These results are summarized in Table~\ref{tab:Summary}. 

Finally, we present improved lower bounds on the spanning ratio of these graphs (see Section~\ref{sec:LowerBounds}) and we provide a partial order on these families (see Section~\ref{sec:Comparison}). In particular, we show that $\theta$-graphs with $4k + 4$ cones have spanning ratio at least $1 + 2 \tan(\theta/2) + 2 \tan^2(\theta/2)$. This is somewhat surprising since, for equal values of $k$, the spanning ratio of $\theta$-graphs with $4k + 4$ cones is greater than that of $\theta$-graphs with $4k + 2$ cones, showing that increasing the number of cones can make the spanning ratio worse.

\begin{table}[ht]
  \begin{center}
    \begin{tabular}{| >{\centering\arraybackslash}m{\dimexpr.17\linewidth-2\tabcolsep} || >{\centering\arraybackslash}m{\dimexpr.23\linewidth-2\tabcolsep} | >{\centering\arraybackslash}m{\dimexpr.23\linewidth-2\tabcolsep} | >{\centering\arraybackslash}m{\dimexpr.36\linewidth-2\tabcolsep} |}
    \hline
    & Current Spanning & Current Routing & Previous Spanning \& Routing \\ 
    \hline \hline
    \graph{2} & $1 + 2 \sin \left( \frac{\theta}{2} \right)$ & $\frac{1}{1 - 2 \sin \left( \frac{\theta}{2} \right)}$~\cite{RS91} & \vspace{2ex} $\frac{1}{1 - 2 \sin \left( \frac{\theta}{2} \right)}$~\cite{RS91} \\ [2ex]
    \hline
    \graph{3} & $\frac{\cos \left( \frac{\theta}{4} \right)}{\cos \left( \frac{\theta}{2} \right) - \sin \left( \frac{3\theta}{4} \right)}$ & $1 + \frac{2 \sin \left( \frac{\theta}{2} \right) \cos \left( \frac{\theta}{4} \right)}{\cos \left( \frac{\theta}{2} \right) - \sin \left( \frac{\theta}{2} \right)}$  & \vspace{2ex} $\frac{1}{1 - 2 \sin \left( \frac{\theta}{2} \right)}$~\cite{RS91} \\ [2ex] 
    \hline
    \graph{4} & $1 + \frac{2 \sin \left( \frac{\theta}{2} \right)}{\cos \left( \frac{\theta}{2} \right) - \sin \left( \frac{\theta}{2} \right)}$ & $1 + \frac{2 \sin \left( \frac{\theta}{2} \right)}{\cos \left( \frac{\theta}{2} \right) - \sin \left( \frac{\theta}{2} \right)}$ & \vspace{2ex} $\frac{1}{1 - 2 \sin \left( \frac{\theta}{2} \right)}$~\cite{RS91} \\ [2ex]
    \hline
    \graph{5} & $\frac{\cos \left( \frac{\theta}{4} \right)}{\cos \left( \frac{\theta}{2} \right) - \sin \left( \frac{3\theta}{4} \right)}$ & $1 + \frac{2 \sin \left( \frac{\theta}{2} \right) \cos \left( \frac{\theta}{4} \right)}{\cos \left( \frac{\theta}{2} \right) - \sin \left( \frac{\theta}{2} \right)}$  & \vspace{2ex} $\frac{1}{1 - 2 \sin \left( \frac{\theta}{2} \right)}$~\cite{RS91} \\ [2ex] 
    \hline
    \end{tabular}
  \end{center} 
  \caption{An overview of current and previous spanning and routing ratios of $\theta$-graphs}
  \label{tab:Summary}
\end{table}

\section{Preliminaries}
\label{sec:Preliminaries}
Let a \emph{cone} be the region in the plane between two rays originating from the same vertex (referred to as the apex of the cone). When constructing a $\theta_m$-graph, for each vertex $u$ consider the rays originating from~$u$ with the angle between consecutive rays being $\theta = 2 \pi / m$ (see Figure~\ref{fig:Cones}). Each pair of consecutive rays defines a cone. The cones are oriented such that the bisector of some cone coincides with the vertical halfline through $u$ that lies above $u$. We refer to this cone as $C^u_0$ and number the cones in clockwise order around $u$. The cones around the other vertices have the same orientation as the ones around $u$. If the apex is clear from the context, we write $C_i$ to indicate the $i$-th cone. 

For ease of exposition, we only consider point sets in general position: no two vertices lie on a line parallel to one of the rays that define the cones, no two vertices lie on a line perpendicular to the bisector of one of the cones, and no three points are collinear. 

\begin{figure}[ht]
  \begin{center}
    \includegraphics{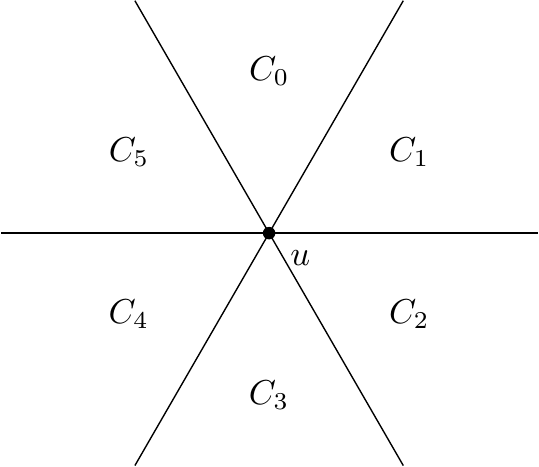}
  \end{center}
  \caption{The cones having apex $u$ in the $\theta_6$-graph}
  \label{fig:Cones}
\end{figure}

The $\theta_m$-graph is constructed as follows: for each cone $C^u_i$ of each vertex~$u$, add an edge from~$u$ to the closest vertex in that cone, where the distance is measured along the bisector of the cone (see Figure~\ref{fig:Projection}). More formally, we add an edge between two vertices $u$ and $v$ if $v \in C^u_i$, and for all vertices $w \in C^u_i$, $|u v'| \leq |u w'|$, where $v'$ and $w'$ denote the orthogonal projection of $v$ and $w$ onto the bisector of $C_i$. Note that our assumptions of general position imply that each vertex adds at most one edge per cone to the graph. 

\begin{figure}[ht]
  \begin{center}
    \includegraphics{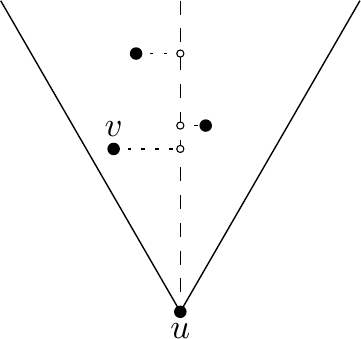}
  \end{center}
  \caption{Three vertices are projected onto the bisector of a cone of $u$. Vertex $v$ is the closest vertex}
  \label{fig:Projection}
\end{figure}

Using the structure of the $\theta_m$-graph, \emph{$\theta$-routing} is defined as follows. Let $t$ be the destination of the routing algorithm and let $u$ be the current vertex. If there exists a direct edge to $t$, follow this edge. Otherwise, follow the edge to the closest vertex in the cone of $u$ that contains $t$. 

Finally, given a vertex $w$ in cone $C$ of a vertex $u$, we define the \emph{canonical triangle} \canon{u}{w} to be the triangle defined by the borders of $C$ and the line through $w$ perpendicular to the bisector of $C$. We use $m$ to denote the midpoint of the side of \canon{u}{w} opposite $u$ and $\alpha$ to denote the smaller unsigned angle between $u w$ and $u m$ (see Figure~\ref{fig:CanonicalTriangle}). Note that for any pair of vertices $u$ and $w$ in the $\theta_m$-graph, there exist two canonical triangles: \canon{u}{w} and \canon{w}{u}. 

\begin{figure}[ht]
  \begin{center}
    \includegraphics{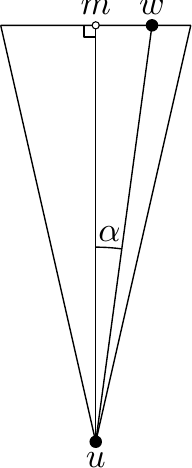}
  \end{center}
  \caption{The canonical triangle \canon{u}{w}}
  \label{fig:CanonicalTriangle}
\end{figure}

\section{Some Geometric Lemmas}
\label{sec:GeometricLemmas}
First, we prove a few geometric lemmas that are useful when bounding the spanning ratios of the graphs. We start with a nice geometric property of the \graph{2}. 

\begin{lemma}
  \label{lem:Boundary}
  In the \graph{2}, any line perpendicular to the bisector of a cone is parallel to the boundary of some cone. 
\end{lemma}
\begin{proof}
  The angle between the bisector of a cone and the boundary of that cone is $\theta/2$. In the \graph{2}, since $\theta = 2\pi / (4 k + 2)$, the angle between the bisector and the line perpendicular to this bisector is $\pi/2 = ((4 k + 2) / 4) \cdot \theta = k \cdot \theta + \theta/2$. Thus the angle between the line perpendicular to the bisector and the boundary of the cone is $\pi/2 - \theta/2 = k \cdot \theta$. Since a cone boundary is placed at every multiple of $\theta$, the line perpendicular to the bisector is parallel to the boundary of some cone. 
\end{proof}

This property helps when bounding the spanning ratio of the \graph{2}. However, before deriving this bound, we prove a few other geometric lemmas. We use $\angle xyz$ to denote the smaller angle between line segments $xy$ and $yz$. 

\begin{lemma}
  \label{lem:FourPoints}
  Let $a$, $b$, $c$, and $d$ be four points on a circle such that $\angle cad \leq \angle bad \leq \angle adc$. It holds that $|a c| + |c d| \leq |a b| + |b d|$ and $|c d| \leq |b d|$. 
\end{lemma}
\begin{proof}
  This situation is illustrated in Figure~\ref{fig:FourPoints}. Without loss of generality, we assume that $|a d| = 1$. Since $b$ and $c$ lie on the same circle and $\angle abd$ and $\angle acd$ are the angle opposite to the same chord $a d$, the inscribed angle theorem implies that $\angle abd = \angle acd$. Furthermore, since $\angle cad \leq \angle adc$, $c$ lies to the right of the perpendicular bisector of $a d$. 

  \begin{figure}
    \begin{center}
      \includegraphics{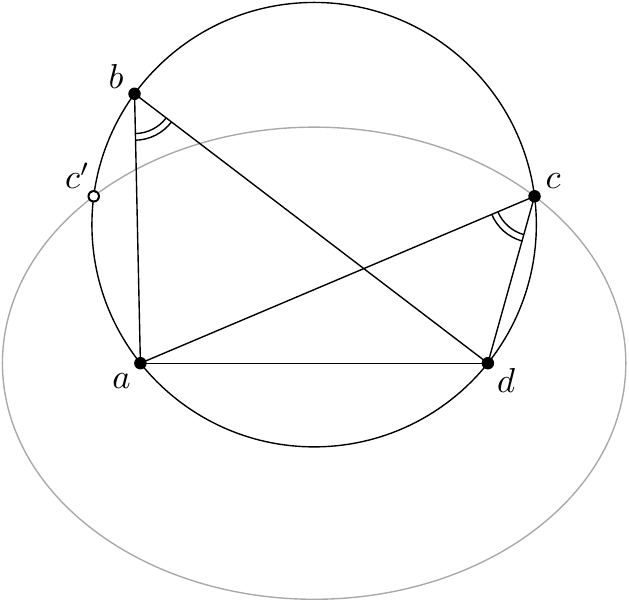}
    \end{center}
    \caption{Illustration of the proof of Lemma~\ref{lem:FourPoints}}
    \label{fig:FourPoints}
  \end{figure}

  First, we show that $|a c| + |c d| \leq |a b| + |b d|$ by showing that $|a c| + |c d| + |a d| \leq |a b| + |b d| + |a d|$. Let $c'$ be the point on the circle when we mirror $c$ along the perpendicular bisector of $a d$. Points $c$ and $c'$ partition the circle into two arcs. Since $\angle cad \leq \angle bad \leq \angle adc$, $b$ lies on the upper arc of the circle. We focus on triangle $a c d$. The locus of the point $c$ such that the perimeter of $a c d$ is constant defines an ellipse. This ellipse has major axis $a d$ and goes through $c$ and $c'$. Since this major axis is horizontal, the ellipse does not intersect the upper arc of the circle. Hence, since $b$ lies on the upper arc of the circle, which is outside of the ellipse, the perimeter of $a b d$ is greater than that of $a c d$, completing the first half of the proof. 

  Next, we show that $|c d| \leq |b d|$. Using the sine law, we have that $|c d| = \sin \angle cad /$ $\sin \angle acd$ and $|b d| = \sin \angle bad / \sin \angle abd$. Since $\angle cad \leq \angle bad \leq \angle adc \leq \pi - \angle cad$, we have that $\sin \angle cad \leq \sin \angle bad$. Hence, since $\angle abd = \angle acd$, we have that $|c d| \leq |b d|$. 
\end{proof}

\begin{lemma}
  \label{lem:ApplyFourPoints} 
  Let $u$, $v$ and $w$ be three vertices in the \graph{x}, where $x \in \{2, 3, 4, 5\}$, such that $w \in C_0^u$ and $v \in \canon{u}{w}$, to the left of $w$. Let $a$ be the intersection of the side of $\canon{u}{w}$ opposite to $u$ with the left boundary of $C_0^v$. Let $C_i^v$ denote the cone of $v$ that contains $w$ and let $c$ and $d$ be the upper and lower corner of $\canon{v}{w}$. If $1 \leq i \leq k-1$, or $i = k$ and $|c w| \leq |d w|$, then $\max \left\{|v c| + |c w|, |v d| + |d w|\right\} \leq |v a| + |a w|$ and $\max \left\{|c w|, |d w|\right\} \leq |a w|$.
\end{lemma}
\begin{proof}
  This situation is illustrated in Figure~\ref{fig:ApplyFourPoints}. We perform case distinction on $\max \{|c w|,$ $|d w|\}$. 

  \begin{figure}[ht]
    \begin{center}
      \includegraphics{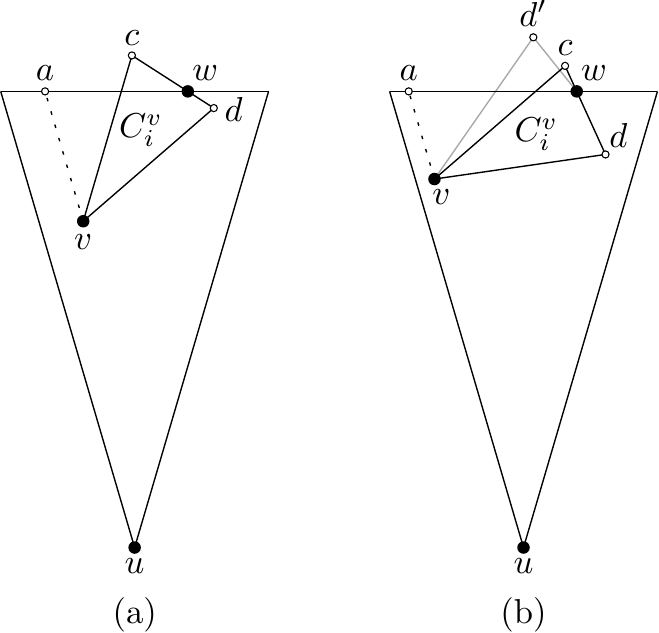}
    \end{center}
    \caption{The two cases for the situation where we apply Lemma~\ref{lem:FourPoints}: (a) $|c w| > |d w|$, \mbox{(b) $|c w| \leq |d w|$}}
    \label{fig:ApplyFourPoints}
  \end{figure}

  \textit{Case 1:} If $|c w| > |d w|$ (see Figure~\ref{fig:ApplyFourPoints}a), we need to show that when $1 \leq i \leq k-1$, we have that $|v c| + |c w| \leq |v a| + |a w|$ and $|c w| \leq |a w|$. Since angles $\angle vaw$ and $\angle vcw$ are both angles between the boundary of a cone and the line perpendicular to its bisector, we have that $\angle vaw = \angle vcw$. Thus, $c$ lies on the circle through $a$, $v$, and $w$. Therefore, if we can show that $\angle cvw \leq \angle avw \leq \angle vwc$, Lemma~\ref{lem:FourPoints} proves this case. 

  We show $\angle cvw \leq \angle avw \leq \angle vwc$ in two steps. Since $w \in C_i^v$ and $i \geq 1$, we have that $\angle avc = i \cdot \theta \geq \theta$. Hence, since $\angle avw = \angle avc + \angle cvw$, we have that $\angle cvw \leq \angle avw$. It remains to show that $\angle avw \leq \angle vwc$. We note that $\angle avw \leq (i + 1) \cdot \theta$ and $(\pi - \theta)/2 \leq \angle vwc$, since $|c w| > |d w|$. Using that $\theta = 2 \pi / (4 k + x)$ and $x \in \{2, 3, 4, 5\}$, we have the following. 
  \begin{eqnarray*}
    i &\leq& k - 1 \\
    i &\leq& k + \frac{x}{4} - \frac{3}{2} \\
    i &\leq& \frac{\pi \cdot (4 k + x)}{4 \pi} - \frac{3}{2} \\
    i &\leq& \frac{\pi}{2 \theta} - \frac{3}{2} \\
    (i + 1) \cdot \theta &\leq& \frac{\pi - \theta}{2} \\
    \angle avw &\leq& \angle vwc 
  \end{eqnarray*}

  \textit{Case 2:} If $|c w| \leq |d w|$ (see Figure~\ref{fig:ApplyFourPoints}b), we need to show that when $1 \leq i \leq k$, we have that $|v d| + |d w| \leq |v a| + |a w|$ and $|d w| \leq |a w|$. Since angles $\angle vaw$ and $\angle vdw$ are both angles between the boundary of a cone and the line perpendicular to its bisector, we have that $\angle vaw = \angle vdw$. Thus, when we reflect $d$ in the line through $v w$, the resulting point $d'$ lies on the circle through $a$, $v$, and $w$. Therefore, if we can show that $\angle d'vw \leq \angle avw \leq \angle vwd'$, Lemma~\ref{lem:FourPoints} proves this case. 

  We show $\angle d'vw \leq \angle avw \leq \angle vwd'$ in two steps. Since $w \in C_i^v$ and $i \geq 1$, we have that $\angle avw \geq \angle avc = i \cdot \theta \geq \theta$. Hence, since $\angle d'vw \leq \theta$, we have that $\angle d'vw \leq \angle avw$. It remains to show that $\angle avw \leq \angle vwd'$. We note that $\angle vwd' = \angle dwv = \pi - (\pi - \theta)/2 - \angle dvw$ and $\angle avw = \angle avd - \angle dvw = (i + 1) \cdot \theta - \angle dvw$. Using that $\theta = 2 \pi / (4 k + x)$ and $x \in \{2, 3, 4, 5\}$, we have the following. 
  \begin{eqnarray*}
    i &\leq& k \\
    i &\leq& k + \frac{x}{4} - \frac{1}{2} \\
    i &\leq& \frac{\pi \cdot (4 k + x)}{4 \pi} - \frac{1}{2} \\
    i &\leq& \frac{\pi}{2 \theta} - \frac{1}{2} \\
    (i + 1) \cdot \theta - \angle dvw &\leq& \frac{\pi + \theta}{2} - \angle dvw \\
    \angle avw &\leq& \angle vwd' 
  \end{eqnarray*} 
\end{proof}

\begin{lemma}
  \label{lem:CalculationCase}
  Let $u$, $v$ and $w$ be three vertices in the \graph{x}, such that $w \in C_0^u$, $v \in \canon{u}{w}$ to the left of $w$, and $w \not \in C_0^v$. Let $a$ be the intersection of the side of $\canon{u}{w}$ opposite to $u$ with the left boundary of $C_0^v$. Let $c$ and $d$ be the corners of $\canon{v}{w}$ opposite to $v$. Let $\beta = \angle a w v$ and let $\gamma$ be the unsigned angle between $v w$ and the bisector of \canon{v}{w}. Let \const be a positive constant. If 
  \begin{align}
  \label{ineq:CalculationCase1}
  \const \geq \frac{\cos \gamma - \sin \beta}{\cos \left( \frac{\theta}{2} - \beta \right) - \sin \left( \frac{\theta}{2} + \gamma \right)},
  \end{align}
  then 
  \begin{align}
  \label{ineq:CalculationCase2}
  \max \left\{|v c| + \const \cdot |c w|, |v d| + \const \cdot |d w|\right\} \leq |v a| + \const \cdot |a w|.
  \end{align}
\end{lemma}
\begin{proof}
  This situation is illustrated in Figure~\ref{fig:CalculationLemma}. Since the angle between the bisector of a cone and its boundary is $\theta/2$, by the sine law, we have the following. 
  \begin{eqnarray*}
    |v c| = |v d| &=& |v w| \cdot \frac{\cos \gamma}{\cos \left( \frac{\theta}{2} \right) } \\ 
    \max \left\{|c w|, |d w| \right\} &=& |v w| \cdot \left( \sin \gamma + \cos \gamma \tan \left( \frac{\theta}{2} \right) \right) \\ 
    |v a| &=& |v w| \cdot \frac{\sin \beta}{\cos \left( \frac{\theta}{2} \right) } \\ 
    |a w| &=& |v w| \cdot \left( \cos \beta + \sin \beta \tan \left( \frac{\theta}{2} \right) \right)
  \end{eqnarray*}

  \begin{figure}[ht]
    \begin{center}
      \includegraphics{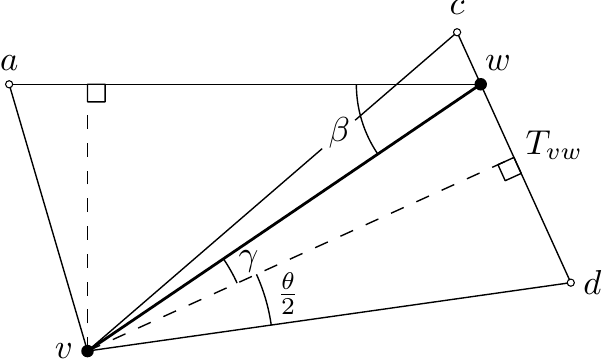}
    \end{center}
    \caption{Finding a constant \const such that $|v d| + \const \cdot |d w| \leq |v a| + \const \cdot |a w|$}
    \label{fig:CalculationLemma}
  \end{figure}

  \noindent To show that (\ref{ineq:CalculationCase2}) holds, we first multiply both sides by $\cos (\theta/2) / |v w|$ and rewrite as follows. 
  \begin{eqnarray*}
    \frac{\cos \left( \frac{\theta}{2} \right)}{|v w|} \cdot \max \left\{|v c| + \const \cdot |c w|, |v d| + \const \cdot |d w|\right\}&& \\
    &\hspace{-6.5cm} =& \hspace{-3.1cm} \cos \gamma + \const \cdot \left( \sin \gamma \cos \left( \frac{\theta}{2} \right) + \cos \gamma \sin \left( \frac{\theta}{2} \right) \right) \\
    &\hspace{-6.5cm} =& \hspace{-3.1cm} \cos \gamma + \const \cdot \sin \left( \frac{\theta}{2} + \gamma \right)
  \end{eqnarray*}
  \begin{eqnarray*}
    \frac{\cos \left( \frac{\theta}{2} \right)}{|v w|} \cdot (|v a| + \const \cdot |a w|) &=& \sin \beta + \const \cdot \left( \cos \beta \cos \left( \frac{\theta}{2} \right) + \sin \beta \sin \left( \frac{\theta}{2} \right) \right) \\
    &=& \sin \beta + \const \cdot \cos \left( \frac{\theta}{2} - \beta \right)
  \end{eqnarray*}
  Therefore, to prove that (\ref{ineq:CalculationCase1}) implies (\ref{ineq:CalculationCase2}), we rewrite (\ref{ineq:CalculationCase1}) as follows. 
  \begin{eqnarray*}
    \const &\geq& \frac{\cos \gamma - \sin \beta}{\cos \left( \frac{\theta}{2} - \beta \right) - \sin \left( \frac{\theta}{2} + \gamma \right)} \\
    \cos \gamma - \sin \beta &\leq& \const \cdot \left( \cos \left( \frac{\theta}{2} - \beta \right) - \sin \left( \frac{\theta}{2} + \gamma \right) \right) \\
    \cos \gamma + \const \cdot \sin \left( \frac{\theta}{2} + \gamma \right) &\leq& \sin \beta + \const \cdot \cos \left( \frac{\theta}{2} - \beta \right) 
  \end{eqnarray*} 

  It remains to show that $\const > 0$. Since $w \not \in C_0^v$, we have that \mbox{$\beta \in (0, (\pi - \theta)/2)$}. Moreover, we have that $\gamma \in [0, \theta/2)$, by definition. This implies that $\sin (\pi/2 + \gamma) > \sin \beta$, or equivalently, $\cos \gamma - \sin \beta > 0$. Thus, we need to show that $\cos (\theta/2 - \beta) - \sin (\theta/2 + \gamma) > 0$, or equivalently, $\sin (\pi/2 + \theta/2 - \beta) > \sin (\theta/2 + \gamma)$. It suffices to show that $\theta/2 + \gamma < \pi/2 + \theta/2 - \beta < \pi - \theta/2 - \gamma$. This follows from $\beta \in (0, (\pi - \theta)/2)$, $\gamma \in [0, \theta/2)$, and the fact that $\theta \leq 2\pi/7$. 
\end{proof}

\section{Upper Bounds}
In this section, we provide improved upper bounds for the four families of $\theta$-graphs: the \graph{2}, the \graph{3}, the \graph{4}, and the \graph{5}. We first prove that the \graph{2} has a tight spanning ratio of $1 + 2 \sin (\theta/2)$. Next, we provide a generic framework for the spanning proof for the three other families of $\theta$-graphs. After providing this framework, we fill in the blanks for the individual families.

\subsection{Optimal Bounds on the \Graph{2}}
\label{subsec:Theta4k+2}
We start by showing that the \graph{2} has a spanning ratio of $1 + 2 \sin (\theta/2)$. At the end of this section, we also provide a matching lower bound, proving that this spanning ratio is tight. 

\begin{theorem}
  \label{theo:PathLength}
  Let $u$ and $w$ be two vertices in the plane. Let $m$ be the midpoint of the side of \canon{u}{w} opposite $u$ and let $\alpha$ be the unsigned angle between $u w$ and $u m$. There exists a path connecting $u$ and $w$ in the \graph{2} of length at most \[\left( \left(\frac{1 + \sin \left(\frac{\theta}{2}\right)}{\cos \left(\frac{\theta}{2}\right)} \right) \cdot \cos \alpha + \sin \alpha \right) \cdot |u w|.\] 
\end{theorem}
\begin{proof}
  We assume without loss of generality that $w \in C_0^u$. We prove the theorem by induction on the area of $\canon{u}{w}$ (formally, induction on the rank, when ordered by area, of the canonical triangles for all pairs of vertices). Let $a$ and $b$ be the upper left and right corners of $\canon{u}{w}$ and let $y$ and $z$ be the left and right intersections of the left and right boundaries of $\canon{u}{w}$ and the boundaries of $C_{2k+1}^w$, the cone of $w$ that contains $u$ (see Figure~\ref{fig:TriangleCases}). Our inductive hypothesis is the following, where $\delta(u,w)$ denotes the length of the shortest path from $u$ to $w$ in the \graph{2}:
  \begin{itemize}
    \item If $a y w$ is empty, then $\delta(u, w) \leq |u b| + |b w|$.
    \item If $b z w$ is empty, then $\delta(u, w) \leq |u a| + |a w|$.
    \item If neither $a y w$ nor $b z w$ is empty, then $\delta(u, w) \leq \max\{|u a| + |a w|, |u b| + |b w|\}$. 
  \end{itemize}
  Note that if both $a y w$ and $b z w$ are empty, the induction hypothesis implies that $\delta(u, w) \leq \min\{|u a| + |a w|, |u b| + |b w|\}$. 

  We first show that this induction hypothesis implies the theorem. Basic trigonometry gives us the following equalities: $|u m| = |u w| \cdot \cos \alpha$, $|m w| = |u w| \cdot \sin \alpha$, $|a m| = |b m| = |u w| \cdot \cos \alpha \tan (\theta/2)$, and $|u a| = |u b| = |u w| \cdot \cos \alpha / \cos (\theta/2)$. Thus, the induction hypothesis gives us that \[\delta(u, w) \leq|u a| + |a m| + |m w| = \left( \left(\frac{1 + \sin \left(\frac{\theta}{2}\right)}{\cos \left(\frac{\theta}{2}\right)} \right) \cdot \cos \alpha + \sin \alpha \right) \cdot |u w|.\] 

  \textbf{Base case:} $\canon{u}{w}$ has rank 1. Since the triangle is a smallest triangle, $w$ is the closest vertex to $u$ in that cone. Hence, the edge $(u,w)$ is part of the \graph{2} and $\delta(u, w) = |u w|$. From the triangle inequality, we have $|u w| \leq \min\{|u a| + |a w|, |u b| + |b w|\}$, so the induction hypothesis holds.

  \textbf{Induction step:} We assume that the induction hypothesis holds for all pairs of vertices with canonical triangles of rank up to $j$. Let $\canon{u}{w}$ be a canonical triangle of rank $j+1$.

  If $(u,w)$ is an edge in the \graph{2}, the induction hypothesis follows from the same argument as in the base case. If there is no edge between $u$ and $w$, let $v$ be the vertex closest to $u$ in $C_0^u$, and let $a'$ and $b'$ be the upper left and right corners of $\canon{u}{v}$ (see Figure~\ref{fig:TriangleCases}). By definition, $\delta(u, w) \leq |u v| + \delta(v, w)$, and by the triangle inequality, $|u v| \leq \min\{|u a'| + |a' v|, |u b'| + |b' v|\}$.

  \begin{figure}[ht]
    \begin{center}
      \includegraphics{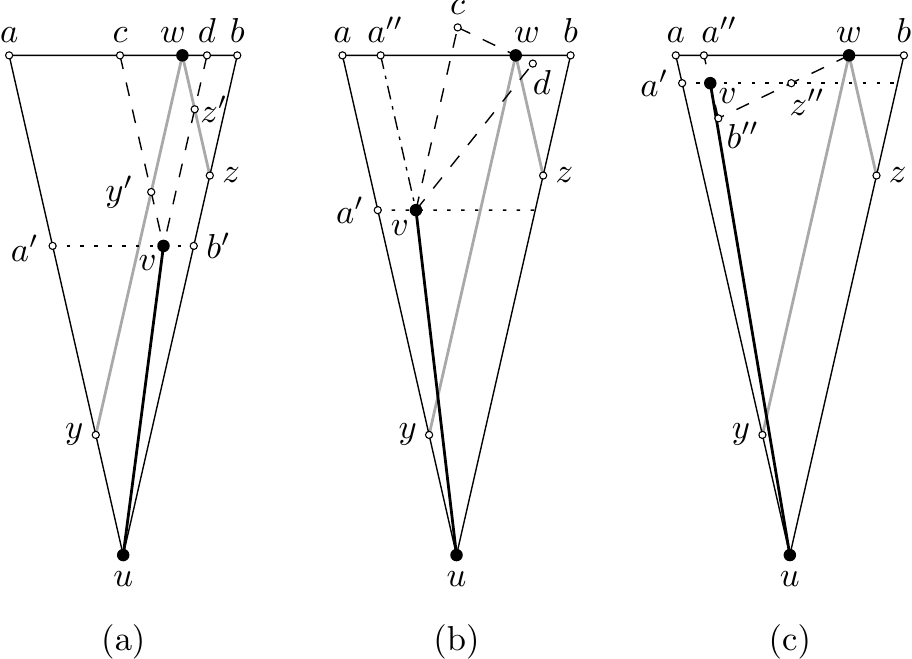}
    \end{center}
    \caption{The three cases of the induction step based on the cone of $v$ that contains $w$, in this case for the $\theta_{14}$-graph}
    \label{fig:TriangleCases}
  \end{figure}

  Without loss of generality, we assume that $v$ lies to the left of $w$. We perform a case analysis based on the cone of $v$ that contains $w$: (a) $w \in C_0^v$, (b) $w \in C_i^v$ where $1 \leq i \leq k-1$, (c) $w \in C_k^v$. 

  \textbf{Case (a):} Vertex $w$ lies in $C_0^v$ (see Figure~\ref{fig:TriangleCases}a). Let $c$ and $d$ be the upper left and right corners of $\canon{v}{w}$, and let $y'$ and $z'$ be the left and right intersections of $\canon{v}{w}$ and the boundaries of $C_{2k+1}^w$. Since $\canon{v}{w}$ has smaller area than $\canon{u}{w}$, we apply the inductive hypothesis to $\canon{v}{w}$. We need to prove all three statements of the inductive hypothesis for $\canon{u}{w}$.

  \begin{enumerate}
  \item If $a y w$ is empty, then $c y' w$ is also empty, so by induction $\delta(v, w) \leq |v d| + |d w|$. Since $v$, $d$, $b$, and $b'$ form a parallelogram, we have:
  \begin{eqnarray*}
    \delta(u, w) &\leq& |u v| + \delta(v, w) \\
    &\leq& |u b'| + |b' v| + |v d| + |d w| \\
    &=& |u b| + |b w|,
  \end{eqnarray*}
  which proves the first statement of the induction hypothesis.

  \item If $b z w$ is empty, an analogous argument proves the second statement of the induction hypothesis.

  \item If neither $a y w$ nor $b z w$ is empty, by induction we have $\delta(v, w) \leq \max\{|v c| + |c w|, |v d| + |d w|\}$. Assume, without loss of generality, that the maximum of the right hand side is attained by its second argument $|v d| + |d w|$ (the other case is similar). Since vertices $v$, $d$, $b$, and $b'$ form a parallelogram, we have that:
  \begin{eqnarray*}
    \delta(u, w) &\leq& |uv| + \delta(v, w) \\
    &\leq& |u b'| + |b' v| +  |v d| + |d w| \\
    &\leq& |u b| + |b w| \\
    &\leq& \max\{|u a| + |a w|, |u b| + |b w|\},
  \end{eqnarray*}
  which proves the third statement of the induction hypothesis. 
  \end{enumerate}

  \textbf{Case (b):} Vertex $w$ lies in $C_i^v$ where $1 \leq i \leq k-1$ (see Figure~\ref{fig:TriangleCases}b). In this case, $v$ lies in $a y w$. Therefore, the first statement of the induction hypothesis for \canon{u}{w} is vacuously true. It remains to prove the second and third statement of the induction hypothesis. Let $a''$ be the intersection of the side of $\canon{u}{w}$ opposite $u$ and the left boundary of $C_0^v$. Since $\canon{v}{w}$ is smaller than $\canon{u}{w}$, by induction we have $\delta(v, w) \leq \max\{|v c| + |c w|, |v d| + |d w|\}$. Since $w \in C_i^v$ where $1 \leq i \leq k-1$, we can apply Lemma~\ref{lem:ApplyFourPoints}. Note that point $a$ in Lemma~\ref{lem:ApplyFourPoints} corresponds to point $a''$ in this proof. Hence, we get that $\max \{|vc| + |cw|, |vd| + |dw| \} \leq |va''| + |a''w|$. Since $|u v| \leq |u a'| + |a' v|$ and $v$, $a''$, $a$, and $a'$ form a parallelogram, we have that $\delta(u, w) \leq |u a| + |a w|$, proving the induction hypothesis for \canon{u}{w}.

  \textbf{Case (c):} Vertex $w$ lies in $C_k^v$ (see Figure~\ref{fig:TriangleCases}c). Since $v$ lies in $a y w$, the first statement of the induction hypothesis for \canon{u}{w} is vacuously true. It remains to prove the second and third statement of the induction hypothesis. Let $a''$ and $b''$ be the upper and lower left corners of \canon{w}{v}, and let $z''$ be the intersection of \canon{w}{v} and the lower boundary of $C_k^v$, i.e. the cone of $v$ that contains $w$. Note that $z''$ is also the right intersection of \canon{u}{v} and \canon{w}{v}. Since $v$ is the closest vertex to $u$, \canon{u}{v} is empty. Hence, $b'' z'' v$ is empty. Since \canon{w}{v} is smaller than \canon{u}{w}, we can apply induction on it. As $b'' z'' v$ is empty, the induction hypothesis for \canon{w}{v} gives $\delta(v, w) \leq |v a''| + |a'' w|$. Since $|u v| \leq |u a'| + |a' v|$ and $v$, $a''$, $a$, and $a'$ form a parallelogram, we have that $\delta(u, w) \leq |u a| + |a w|$, proving the second and third statement of the induction hypothesis for \canon{u}{w}. 
\end{proof}

Since $((1 + \sin(\theta/2)) / \cos(\theta/2)) \cdot \cos \alpha + \sin \alpha$ is increasing for $\alpha \in [0, \theta/2]$, for $\theta \leq \pi/3$, it is maximized when $\alpha = \theta/2$, and we obtain the following corollary: 

\begin{coro}
  \label{cor:SpanningRatio}
  The \graph{2} is a $\left( 1 + 2 \sin \left( \theta/2 \right) \right)$-spanner. 
\end{coro}

The upper bounds given in Theorem~\ref{theo:PathLength} and Corollary~\ref{cor:SpanningRatio} are tight, as shown in Figure~\ref{fig:RatioTight}: we place a vertex~$v$ arbitrarily close to the upper corner of \canon{u}{w} that is furthest from $w$. Likewise, we place a vertex $v'$ arbitrarily close to the lower corner of \canon{w}{u} that is furthest from $u$. Both shortest paths between $u$ and $w$ visit either $v$ or $v'$, so the path length is arbitrarily close to $(((1 + \sin(\theta/2))/\cos(\theta/2)) \cdot \cos \alpha + \sin \alpha) \cdot |u w|$, showing that the upper bounds are tight. 

\begin{figure}[ht]
  \begin{center}
    \includegraphics{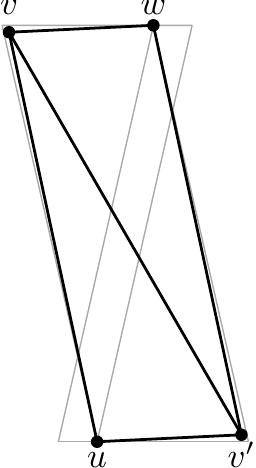}
  \end{center}
  \caption{The lower bound for the \graph{2}}
  \label{fig:RatioTight}
\end{figure}

\subsection{Generic Framework for the Spanning Proof}
In this section, we provide a generic framework for the spanning proof for the three other families of $\theta$-graphs: the \graph{3}, the \graph{4}, and the \graph{5}. This framework contains those parts of the spanning proof that are identical for all three families. In the subsequent sections, we handle the single case that depends on each specific family and determines their respective spanning ratios. 

\begin{theorem}
  \label{theo:PathLengthGeneric}
  Let $u$ and $w$ be two vertices in the plane. Let $m$ be the midpoint of the side of \canon{u}{w} opposite $u$ and let $\alpha$ be the unsigned angle between $u w$ and $u m$. There exists a path connecting $u$ and $w$ in the \graph{x} of length at most 
  \[\left( \frac{\cos \alpha}{\cos \left(\frac{\theta}{2}\right)} + \const \cdot \left(\cos \alpha \tan \left(\frac{\theta}{2}\right) + \sin \alpha\right) \right) \cdot |u w|,\] where $\const \geq 1$ is a function that depends on $x \in \{3, 4, 5\}$ and $\theta$. For the \graph{4}, \const equals $1 / (\cos (\theta/2) - \sin (\theta/2))$ and for the \graph{3} and \graph{5}, \const equals $\cos (\theta/4) /$ $(\cos (\theta/2) - \sin (3\theta/4))$.
\end{theorem}
\begin{proof}
  We assume without loss of generality that $w \in C_0^u$. We prove the theorem by induction on the area of \canon{u}{w} (formally, induction on the rank, when ordered by area, of the canonical triangles for all pairs of vertices). Let $a$ and $b$ be the upper left and right corners of \canon{u}{w}. Our inductive hypothesis is the following, where $\delta(u,w)$ denotes the length of the shortest path from $u$ to $w$ in the \graph{x}: $\delta(u, w) \leq \max\{|u a| + \const \cdot |a w|, |u b| + \const \cdot |b w|\}$. 

  We first show that this induction hypothesis implies the theorem. Basic trigonometry gives us the following equalities: $|u m| = |u w| \cdot \cos \alpha$, $|m w| = |u w| \cdot \sin \alpha$, $|a m| = |b m| = |u w| \cdot \cos \alpha \tan (\theta/2)$, and $|u a| = |u b| = |u w| \cdot \cos \alpha / \cos (\theta/2)$. Thus the induction hypothesis gives that \[\delta(u, w) \leq |u a| + \const \cdot (|a m| + |m w|) = \left( \frac{\cos \alpha}{\cos \left(\frac{\theta}{2}\right)} + \const \cdot \left(\cos \alpha \tan \left(\frac{\theta}{2}\right) + \sin \alpha\right) \right) \cdot |u w|.\] 

  \textbf{Base case:} \canon{u}{w} has rank 1. Since the triangle is a smallest triangle, $w$ is the closest vertex to $u$ in that cone. Hence, the edge $(u,w)$ is part of the \graph{x} and $\delta(u, w) = |u w|$. From the triangle inequality and the fact that $\const \geq 1$, we have $|u w| \leq \max\{|u a| + \const \cdot |a w|, |u b| + \const \cdot |b w|\}$, so the induction hypothesis holds.

  \textbf{Induction step:} We assume that the induction hypothesis holds for all pairs of vertices with canonical triangles of rank up to $j$. Let \canon{u}{w} be a canonical triangle of rank $j+1$.

  If $(u,w)$ is an edge in the \graph{x}, the induction hypothesis follows from the same argument as in the base case. If there is no edge between $u$ and $w$, let $v$ be the vertex closest to $u$ in \canon{u}{w}, and let $a'$ and $b'$ be the upper left and right corners of \canon{u}{v} (see Figure~\ref{fig:TriangleCasesGeneric}). By definition, $\delta(u, w) \leq |u v| + \delta(v, w)$, and by the triangle inequality, $|u v| \leq \min\{|u a'| + |a' v|, |u b'| + |b' v|\}$.

  \begin{figure}[ht]
    \begin{center}
      \includegraphics{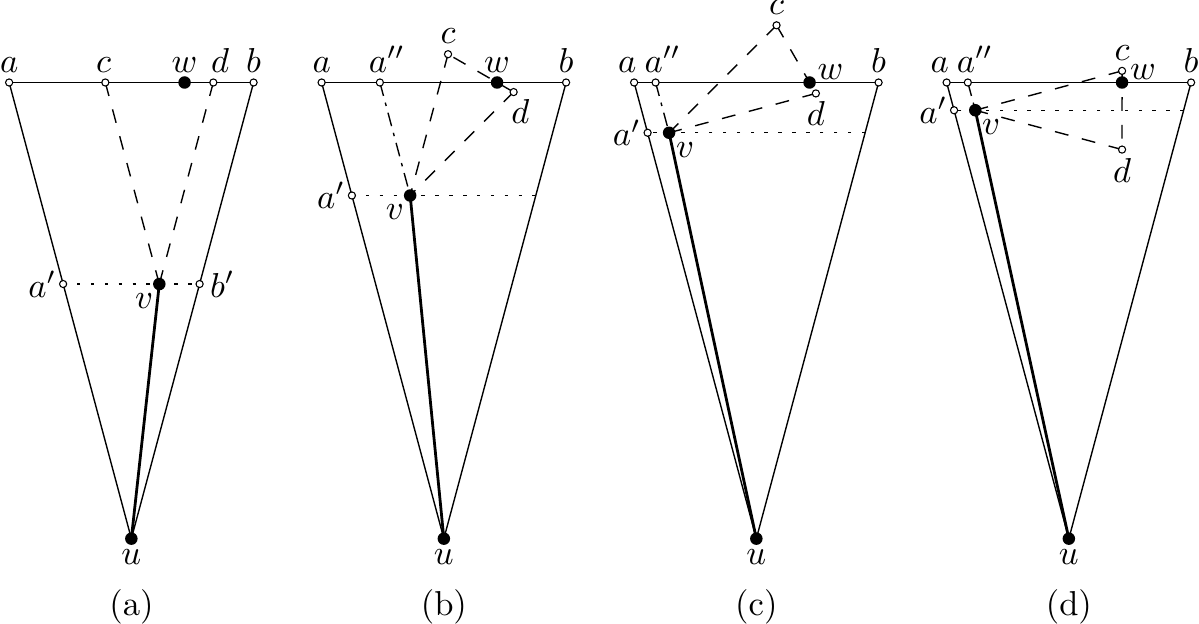}
    \end{center}
    \caption{The four cases of the induction step based on the cone of $v$ that contains $w$, in this case for the $\theta_{12}$-graph}
    \label{fig:TriangleCasesGeneric}
  \end{figure}

  Without loss of generality, we assume that $v$ lies to the left of $w$. We perform a case analysis based on the cone of $v$ that contains $w$, where $c$ and $d$ are the left and right corners of \canon{v}{w}, opposite to $v$: (a) $w \in C_0^v$, (b) $w \in C_i^v$ where $1 \leq i \leq k-1$, or $i = k$ and $|c w| \leq |d w|$, (c) $w \in C_k^v$ and $|c w| > |d w|$, \mbox{(d) $w \in C_{k+1}^v$}. 

  \textbf{Case (a):} Vertex $w$ lies in $C_0^v$ (see Figure~\ref{fig:TriangleCasesGeneric}a). Since \canon{v}{w} has smaller area than \canon{u}{w}, we apply the inductive hypothesis to \canon{v}{w}. Hence we have $\delta(v, w) \leq \max\{|v c| + \const \cdot |c w|, |v d| + \const \cdot |d w|\}$. Since $v$ lies to the left of $w$, the maximum of the right hand side is attained by its first argument, $|v c| + \const \cdot |c w|$. Since vertices $v$, $c$, $a$, and $a'$ form a parallelogram, and $\const \geq 1$, we have that
  \begin{eqnarray*}
    \delta(u, w) &\leq& |uv| + \delta(v, w) \\
		 &\leq& |u a'| + |a' v| + |v c| + \const \cdot |c w| \\
		 &\leq& |u a| + \const \cdot |a w| \\
		 &\leq& \max\{|u a| + \const \cdot |a w|, |u b| + \const \cdot |b w|\},
  \end{eqnarray*}
  which proves the induction hypothesis. 

  \textbf{Case (b):} Vertex $w$ lies in $C_i^v$, where $1 \leq i \leq k-1$, or $i = k$ and $|c w| \leq |d w|$ (see Figure~\ref{fig:TriangleCasesGeneric}b). Let $a''$ be the intersection of the side of $\canon{u}{w}$ opposite $u$ and the left boundary of $C_0^v$. Since $\canon{v}{w}$ is smaller than $\canon{u}{w}$, by induction we have $\delta(v, w) \leq \max\{|v c| + \const \cdot |c w|, |v d| + \const \cdot |d w|\}$. Since $w \in C_i^v$ where $1 \leq i \leq k-1$, or $i = k$ and $|c w| \leq |d w|$, we can apply Lemma~\ref{lem:ApplyFourPoints}. Note that point $a$ in Lemma~\ref{lem:ApplyFourPoints} corresponds to point $a''$ in this proof. Hence, we get that $\max\left\{|vc| + |cw|, |vd| + |dw|\right\} \leq |va''| + |a''w|$ and $\max\left\{|cw|, |dw|\right\} \leq |a''w|$. Since $\const \geq 1$, this implies that $\max \{ |vc| + \const \cdot |cw|,$ $|vd| + \const \cdot |dw| \} \leq |va''| + \const \cdot |a''w|$. Since $|u v| \leq |u a'| + |a' v|$ and $v$, $a''$, $a$, and $a'$ form a parallelogram, we have that $\delta(u, w) \leq |u a| + \const \cdot |a w|$, proving the induction hypothesis for \canon{u}{w}.

  \textbf{Case (c) and (d)} Vertex $w$ lies in $C_k^v$ and $|c w| > |d w|$, or $w$ lies in $C_{k+1}^v$ (see Figures~\ref{fig:TriangleCasesGeneric}c and d). Let $a''$ be the intersection of the side of $\canon{u}{w}$ opposite $u$ and the left boundary of $C_0^v$. Since \canon{v}{w} is smaller than \canon{u}{w}, we can apply induction on it. The actual application of the induction hypothesis varies for the three families of $\theta$-graphs and, using Lemma~\ref{lem:CalculationCase}, determines the value of \const. Hence, these cases are discussed in the spanning proofs of the three families. 
\end{proof}

\subsection{Upper Bound on the \Graph{4}}
\label{subsec:Theta4k+4}
In this section, we improve the upper bounds on the spanning ratio of the \graph{4}, for any integer $k \geq 1$. 

\begin{theorem}
  \label{theo:PathLength4k+4}
  Let $u$ and $w$ be two vertices in the plane. Let $m$ be the midpoint of the side of \canon{u}{w} opposite $u$ and let $\alpha$ be the unsigned angle between $u w$ and $u m$. There exists a path connecting $u$ and $w$ in the \graph{4} of length at most 
  \[\left( \frac{\cos \alpha}{\cos \left(\frac{\theta}{2}\right)} + \frac{\cos \alpha \tan \left(\frac{\theta}{2}\right) + \sin \alpha}{\cos \left(\frac{\theta}{2}\right) - \sin \left(\frac{\theta}{2}\right)} \right) \cdot |u w|.\] 
\end{theorem}
\begin{proof}
  We apply Theorem~\ref{theo:PathLengthGeneric} using $\const = 1 / (\cos (\theta/2) - \sin (\theta/2))$. It remains to handle Case (c), where $w \in C_k^v$ and $|c w| > |d w|$, and Case (d), where $w \in C_{k+1}^v$. 

  Recall that $c$ and $d$ are the left and right corners of \canon{v}{w}, opposite to $v$, and $a''$ is the intersection of the side of $\canon{u}{w}$ opposite $u$ and the left boundary of $C_0^v$. Let $\beta$ be $\angle a'' w v$ and let $\gamma$ be the angle between $v w$ and the bisector of \canon{v}{w}. Since \canon{v}{w} is smaller than \canon{u}{w}, the induction hypothesis gives an upper bound on $\delta(v, w)$. Since $|u v| \leq |u a'| + |a' v|$ and $v$, $a''$, $a$, and $a'$ form a parallelogram, we need to show that $\delta(v, w) \leq |v a''| + \const \cdot |a'' w|$ for both cases in order to complete the proof. 

  \begin{figure}[ht]
    \begin{center}
      \includegraphics{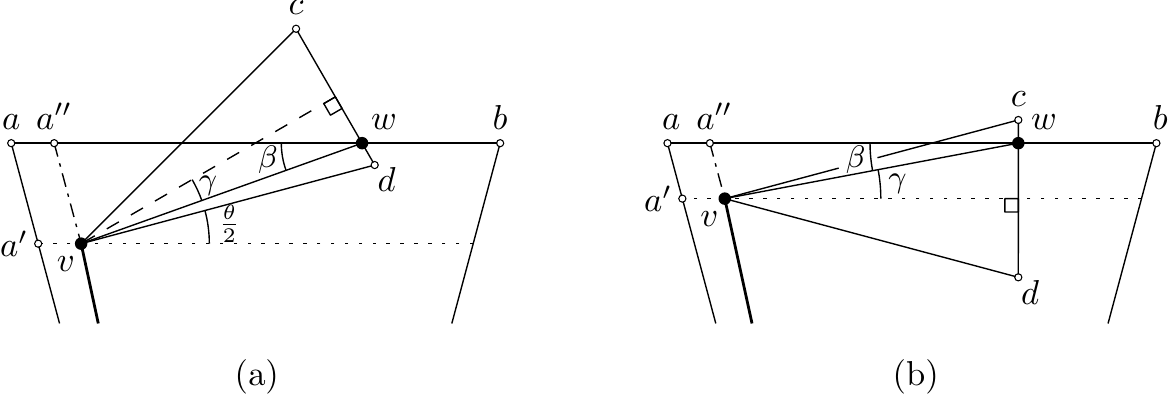}
    \end{center}
    \caption{The remaining cases of the induction step for the \graph{4}: (a) $w$ lies in $C_k^v$ and $|c w| > |d w|$, (b) $w$ lies in $C_{k+1}^v$}
    \label{fig:SpanningProof4k+4}
  \end{figure}

  \textbf{Case (c):} When $w$ lies in $C_k^v$ and $|c w| > |d w|$, the induction hypothesis for \canon{v}{w} gives $\delta(v, w) \leq |v c| + \const \cdot |c w|$ (see Figure~\ref{fig:SpanningProof4k+4}a). We note that $\gamma = \theta - \beta$. Hence, the inequality follows from Lemma~\ref{lem:CalculationCase} when $\const \geq (\cos (\theta - \beta) - \sin \beta) / (\cos (\theta/2 - \beta) - \sin (3\theta/2 - \beta))$. Since this function is decreasing in $\beta$ for $\theta/2 \leq \beta \leq \theta$, it is maximized when $\beta$ equals $\theta/2$. Hence, $\const$ needs to be at least $(\cos (\theta/2) - \sin (\theta/2)) / (1 - \sin \theta)$, which can be rewritten to $1 / (\cos (\theta/2) - \sin (\theta/2))$. 

  \textbf{Case (d):} When $w$ lies in $C_{k+1}^v$, $w$ lies above the bisector of \canon{v}{w} (see Figure~\ref{fig:SpanningProof4k+4}b) and the induction hypothesis for \canon{v}{w} gives $\delta(v, w) \leq |w d| + \const \cdot |d v|$. We note that $\gamma = \beta$. Hence, the inequality follows from Lemma~\ref{lem:CalculationCase} when $\const \geq (\cos \beta - \sin \beta) / (\cos (\theta/2 - \beta) - \sin (\theta/2 + \beta))$, which is equal to $1 / (\cos (\theta/2) - \sin (\theta/2))$. 
\end{proof}

Since $\cos \alpha / \cos (\theta/2) + (\cos \alpha \tan (\theta/2) + \sin \alpha) / (\cos (\theta/2) - \sin (\theta/2))$ is increasing for $\alpha \in [0, \theta/2]$, for $\theta \leq \pi/4$, it is maximized when $\alpha = \theta/2$, and we obtain the following corollary: 

\begin{corollary}
  \label{cor:SpanningRatio4k+4}
  The \graph{4} is a $\left( 1 + \frac{2 \sin \left( \frac{\theta}{2} \right)}{\cos \left( \frac{\theta}{2} \right) - \sin \left( \frac{\theta}{2} \right)} \right)$-spanner. 
\end{corollary}

Furthermore, we observe that the proof of Theorem~\ref{theo:PathLength4k+4} follows the same path as the $\theta$-routing algorithm follows: if the direct edge to the destination is part of the graph, it follows this edge, and if it is not, it follows the edge to the closest vertex in the cone that contains the destination. 

\begin{corollary}
  \label{cor:Routing4k+4}
  The $\theta$-routing algorithm is $\left( 1 + \frac{2 \sin \left( \frac{\theta}{2} \right)}{\cos \left( \frac{\theta}{2} \right) - \sin \left( \frac{\theta}{2} \right)} \right)$-competitive on the \\ \graph{4}. 
\end{corollary}

\subsection{Upper Bounds on the \Graph{3} and \Graph{5}}
\label{subsec:Theta4k+35}
In this section, we improve the upper bounds on the spanning ratio of the \graph{3} and the \graph{5}, for any integer $k \geq 1$. 

\begin{theorem}
  \label{theo:PathLength4k+3}
  Let $u$ and $w$ be two vertices in the plane. Let $m$ be the midpoint of the side of \canon{u}{w} opposite $u$ and let $\alpha$ be the unsigned angle between $u w$ and $u m$. There exists a path connecting $u$ and $w$ in the \graph{3} of length at most 
  \[\left( \frac{\cos \alpha}{\cos \left(\frac{\theta}{2}\right)} + \frac{\left( \cos \alpha \tan \left(\frac{\theta}{2}\right) + \sin \alpha \right) \cdot \cos \left(\frac{\theta}{4}\right)}{\cos \left(\frac{\theta}{2}\right) - \sin \left(\frac{3\theta}{4}\right)} \right) \cdot |u w|.\] 
\end{theorem}
\begin{proof}
  We apply Theorem~\ref{theo:PathLengthGeneric} using $\const = \cos (\theta/4) / (\cos (\theta/2) - \sin (3\theta/4))$. It remains to handle Case (c), where $w \in C_k^v$ and $|c w| > |d w|$, and Case (d), where $w \in C_{k+1}^v$. 

  Recall that $c$ and $d$ are the left and right corners of \canon{v}{w}, opposite to $v$, and $a''$ is the intersection of the side of $\canon{u}{w}$ opposite $u$ and the left boundary of $C_0^v$. Let $\beta$ be $\angle a'' w v$ and let $\gamma$ be the angle between $v w$ and the bisector of \canon{v}{w}. Since \canon{v}{w} is smaller than \canon{u}{w}, the induction hypothesis gives an upper bound on $\delta(v, w)$. Since $|u v| \leq |u a'| + |a' v|$ and $v$, $a''$, $a$, and $a'$ form a parallelogram, we need to show that $\delta(v, w) \leq |v a''| + \const \cdot |a'' w|$ for both cases in order to complete the proof. 

  \begin{figure}[ht]
    \begin{center}
      \includegraphics{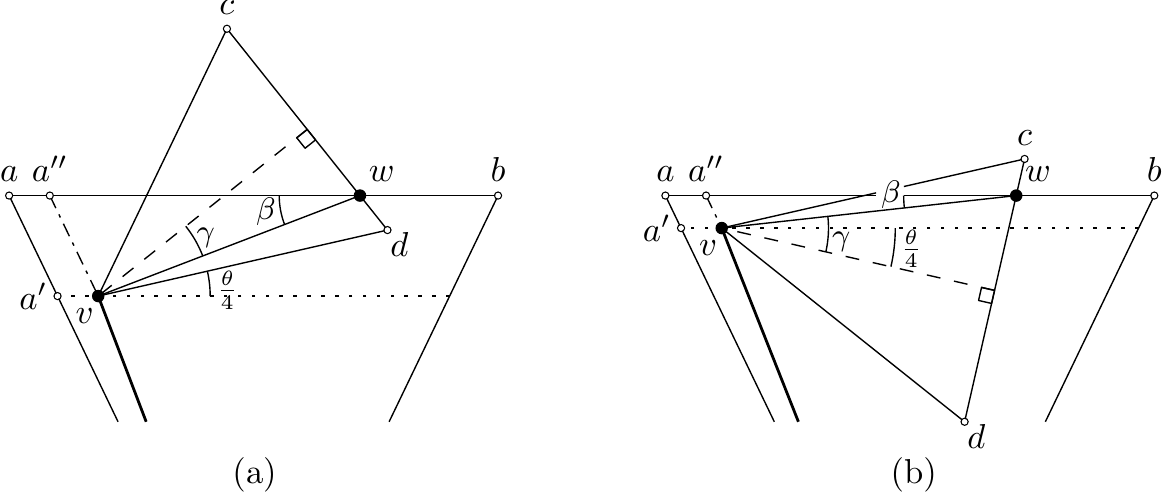}
    \end{center}
    \caption{The remaining cases of the induction step for the \graph{3}: (a) $w$ lies in $C_k^v$ and $|c w| > |d w|$, (b) $w$ lies in $C_{k+1}^v$}
    \label{fig:SpanningProof4k+3}
  \end{figure}

  \textbf{Case (c):} When $w$ lies in $C_k^v$ and $|c w| > |d w|$, the induction hypothesis for \canon{v}{w} gives $\delta(v, w) \leq |v c| + \const \cdot |c w|$ (see Figure~\ref{fig:SpanningProof4k+3}a). We note that $\gamma = 3\theta/4 - \beta$. Hence, the inequality follows from Lemma~\ref{lem:CalculationCase} when $\const \geq (\cos (3\theta/4 - \beta) - \sin \beta) / (\cos (\theta/2 - \beta) - \sin (5\theta/4 - \beta))$. Since this function is decreasing in $\beta$ for $\theta/4 \leq \beta \leq 3\theta/4$, it is maximized when $\beta$ equals $\theta/4$. Hence, $\const$ needs to be at least $(\cos (\theta/2) - \sin (\theta/4)) / (\cos (\theta/4) - \sin \theta)$, which is equal to $\cos (\theta/4) / (\cos (\theta/2) - \sin (3\theta/4))$. 

  \textbf{Case (d):} When $w$ lies in $C_{k+1}^v$, $w$ lies above the bisector of \canon{v}{w} (see Figure~\ref{fig:SpanningProof4k+3}b) and the induction hypothesis for \canon{v}{w} gives $\delta(v, w) \leq |w d| + \const \cdot |d v|$. We note that $\gamma = \theta/4 + \beta$. Hence, the inequality follows from Lemma~\ref{lem:CalculationCase} when $\const \geq (\cos (\theta/4 + \beta) - \sin \beta) / (\cos (\theta/2 - \beta) - \sin (3\theta/4 + \beta))$, which is equal to $\cos (\theta/4) / (\cos (\theta/2) - \sin (3\theta/4))$. 
\end{proof}

\begin{theorem}
  \label{theo:PathLength4k+5}
  Let $u$ and $w$ be two vertices in the plane. Let $m$ be the midpoint of the side of \canon{u}{w} opposite $u$ and let $\alpha$ be the unsigned angle between $u w$ and $u m$. There exists a path connecting $u$ and $w$ in the \graph{5} of length at most 
  \[\left( \frac{\cos \alpha}{\cos \left(\frac{\theta}{2}\right)} + \frac{\left( \cos \alpha \tan \left(\frac{\theta}{2}\right) + \sin \alpha \right) \cdot \cos \left(\frac{\theta}{4}\right)}{\cos \left(\frac{\theta}{2}\right) - \sin \left(\frac{3\theta}{4}\right)} \right) \cdot |u w|.\] 
\end{theorem}
\begin{proof}
  We apply Theorem~\ref{theo:PathLengthGeneric} using $\const = \cos (\theta/4) / (\cos (\theta/2) - \sin (3\theta/4))$. It remains to handle Case (c), where $w \in C_k^v$ and $|c w| > |d w|$, and Case (d), where $w \in C_{k+1}^v$. 

  Recall that $c$ and $d$ are the left and right corners of \canon{v}{w}, opposite to $v$, and $a''$ is the intersection of the side of $\canon{u}{w}$ opposite $u$ and the left boundary of $C_0^v$. Let $\beta$ be $\angle a'' w v$ and let $\gamma$ be the angle between $v w$ and the bisector of \canon{v}{w}. Since \canon{v}{w} is smaller than \canon{u}{w}, the induction hypothesis gives an upper bound on $\delta(v, w)$. Since $|u v| \leq |u a'| + |a' v|$ and $v$, $a''$, $a$, and $a'$ form a parallelogram, we need to show that $\delta(v, w) \leq |v a''| + \const \cdot |a'' w|$ for both cases in order to complete the proof. 

  \begin{figure}[ht]
    \begin{center}
      \includegraphics{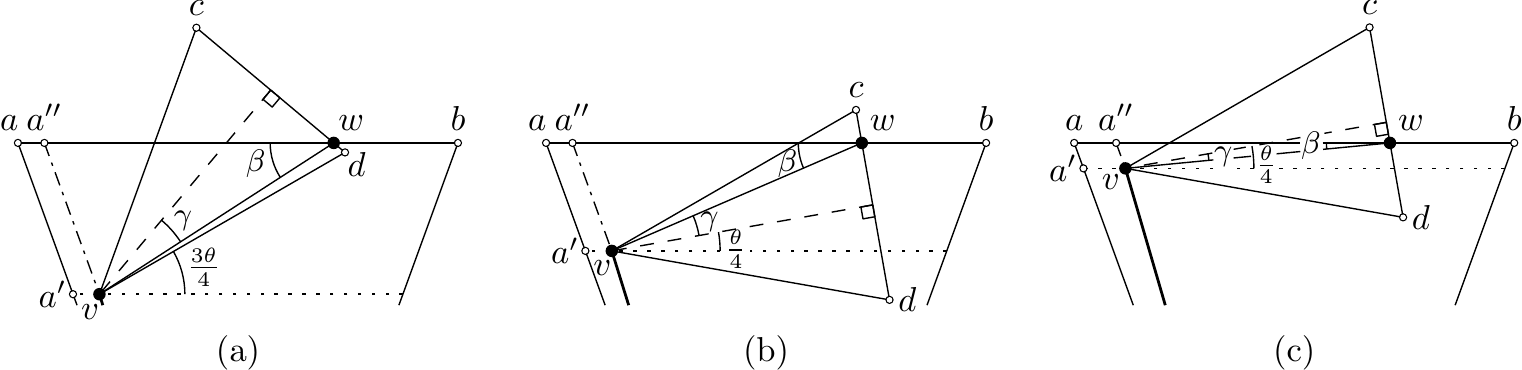}
    \end{center}
    \caption{The remaining cases of the induction step for the \graph{5}: (a) $w$ lies in $C_k^v$ and $|c w| > |d w|$, (b) $w$ lies in $C_{k+1}^v$ and $|c w| < |d w|$, (c) $w$ lies in $C_{k+1}^v$ and $|c w| \geq |d w|$}
    \label{fig:SpanningProof4k+5}
  \end{figure}

  \textbf{Case (c):} When $w$ lies in $C_k^v$ and $|c w| > |d w|$, the induction hypothesis for \canon{v}{w} gives $\delta(v, w) \leq |v c| + \const \cdot |c w|$ (see Figure~\ref{fig:SpanningProof4k+5}a). We note that $\gamma = 5\theta/4 - \beta$. Hence, the inequality follows from Lemma~\ref{lem:CalculationCase} when $\const \geq (\cos (5\theta/4 - \beta) - \sin \beta) / (\cos (\theta/2 - \beta) - \sin (7\theta/4 - \beta))$. Since this function is decreasing in $\beta$ for $3\theta/4 \leq \beta \leq 5\theta/4$, it is maximized when $\beta$ equals $3\theta/4$. Hence, $\const$ needs to be at least $(\cos (\theta/2) - \sin (3\theta/4)) / (\cos (\theta/4) - \sin \theta)$, which is less than $\cos (\theta/4) / (\cos (\theta/2) - \sin (3\theta/4))$. 

  \textbf{Case (d):} When $w$ lies in $C_{k+1}^v$, the induction hypothesis for \canon{v}{w} gives $\delta(v, w) \leq \max\{|v c| + \const \cdot |c w|, |v d| + \const \cdot |d w|\}$. If $|c w| < |d w|$ (see Figure~\ref{fig:SpanningProof4k+5}b), the induction hypothesis for \canon{v}{w} gives $\delta(v, w) \leq |v d| + \const \cdot |d w|$. We note that $\gamma = \beta - \theta/4$. Hence, the inequality follows from Lemma~\ref{lem:CalculationCase} when $\const \geq (\cos (\beta - \theta/4) - \sin \beta) / (\cos (\theta/2 - \beta) - \sin (\theta/4 + \beta))$, which is equal to $\cos (\theta/4) / (\cos (\theta/2) - \sin (3\theta/4))$. 

  If $|c w| \geq |d w|$, the induction hypothesis for \canon{v}{w} gives $\delta(v, w) \leq |v c| + \const \cdot |c w|$ (see Figure~\ref{fig:SpanningProof4k+5}c). We note that $\gamma = \theta/4 - \beta$. Hence, the inequality follows from Lemma~\ref{lem:CalculationCase} when $\const \geq (\cos (\theta/4 - \beta) - \sin \beta) / (\cos (\theta/2 - \beta) - \sin (3\theta/4 - \beta))$. Since this function is decreasing in $\beta$ for $0 \leq \beta \leq \theta/4$, it is maximized when $\beta$ equals $0$. Hence, $\const$ needs to be at least $\cos (\theta/4) / (\cos (\theta/2) - \sin (3\theta/4))$.
\end{proof}

By looking at two vertices $u$ and $w$ in the \graph{3} and the \graph{5}, we can see that when the angle between $u w$ and the bisector of \canon{u}{w} is $\alpha$, the angle between $w u$ and the bisector of \canon{w}{u} is $\theta/2 - \alpha$. Hence the worst case spanning ratio corresponds to the minimum of the spanning ratio when looking at \canon{u}{w} and the spanning ratio when looking at \canon{w}{u}. 

\begin{theorem}
  \label{theo:SpanningRatio4k+3,5}
  The \graph{3} and \graph{5} are $\frac{\cos \left(\frac{\theta}{4}\right)}{\cos \left(\frac{\theta}{2}\right) - \sin \left(\frac{3\theta}{4}\right)}$-spanners. 
\end{theorem}
\begin{proof}
  The spanning ratio of the \graph{3} and the \graph{5} is at most 
  \[ \min \left\{
  \begin{array}{l}
    \frac{\cos \alpha}{\cos \left(\frac{\theta}{2}\right)} + \frac{\left( \cos \alpha \tan \left(\frac{\theta}{2}\right) + \sin \alpha \right) \cdot \cos \left(\frac{\theta}{4}\right)}{\cos \left(\frac{\theta}{2}\right) - \sin \left(\frac{3\theta}{4}\right)}, \\
    \frac{\cos \left(\frac{\theta}{2} - \alpha\right)}{\cos \left(\frac{\theta}{2}\right)} + \frac{\left( \cos \left(\frac{\theta}{2} - \alpha\right) \tan \left(\frac{\theta}{2}\right) + \sin \left(\frac{\theta}{2} - \alpha\right) \right) \cdot \cos \left(\frac{\theta}{4}\right)}{\cos \left(\frac{\theta}{2}\right) - \sin \left(\frac{3\theta}{4}\right)}
  \end{array}
  \right\}.
  \]

  Since $\cos \alpha / \cos (\theta/2) + \const \cdot (\cos \alpha \tan (\theta/2) + \sin \alpha)$ is increasing for $\alpha \in [0, \theta/2]$, for $\theta \leq 2\pi/7$, the minimum of these two functions is maximized when the two functions are equal, i.e. when $\alpha = \theta/4$. Thus the \graph{3} and the \graph{5} have spanning ratio at most 
  \begin{eqnarray*}
    \frac{\cos \left(\frac{\theta}{4}\right)}{\cos \left(\frac{\theta}{2}\right)} + \frac{\left( \cos \left(\frac{\theta}{4}\right) \tan \left(\frac{\theta}{2}\right) + \sin \left(\frac{\theta}{4}\right) \right) \cdot \cos \left(\frac{\theta}{4}\right)}{\cos \left(\frac{\theta}{2}\right) - \sin \left(\frac{3\theta}{4}\right)}
    &=& \frac{\cos \left(\frac{\theta}{4}\right)}{\cos \left(\frac{\theta}{2}\right) - \sin \left(\frac{3\theta}{4}\right)}. 
  \end{eqnarray*} 
\end{proof}

Furthermore, we observe that the proofs of Theorem~\ref{theo:PathLength4k+3} and Theorem~\ref{theo:PathLength4k+5} follow the same path as the $\theta$-routing algorithm follows. Since in the case of routing, we are forced to consider the canonical triangle with the source as apex, the arguments that decreased the spanning ratio cannot be applied. Hence, we obtain the following corollary. 

\begin{corollary}
  \label{cor:Routing4k+3,5}
  The $\theta$-routing algorithm is $\left( 1 + \frac{2 \sin \left(\frac{\theta}{2}\right) \cos \left(\frac{\theta}{4}\right)}{\cos \left(\frac{\theta}{2}\right) - \sin \left(\frac{3\theta}{4}\right)} \right)$-competitive on the \graph{3} and the \graph{5}. 
\end{corollary}

\section{Lower Bounds}
\label{sec:LowerBounds}
In this section, we provide lower bounds for the \graph{3}, the \graph{4}, and the \graph{5}. For each of the families, we construct a lower bound example by extending the shortest path between two vertices $u$ and $w$. For brevity, we describe only how to extend one of the shortest paths between these vertices. To extend all shortest paths between $u$ and $w$, the same transformation is applied to all equivalent paths or canonical triangles. 

For example, when constructing the lower bound for the \graph{3}, our first step is to ensure that there is no edge between $u$ and $w$. To this end, the proof of Theorem~\ref{theo:LowerBound4k+3} states that we place a vertex $v_1$ in the corner of \canon{u}{w} that is furthest from $w$. Placing only this single vertex, however, does not prevent the edge $u w$ from being present, as $u$ is still the closest vertex in \canon{w}{u}. Hence, we also place a vertex in the corner of \canon{w}{u} that is furthest from $u$. Since these two modifications are essentially the same, but applied to different canonical triangles, we describe only the placement of one of these vertices. The full result of each step is shown in the accompanying figures.

\subsection{Lower Bounds on the \Graph{3}}
In this section, we construct a lower bound on the spanning ratio of the \graph{3}, for any integer $k \geq 1$. 

\begin{theorem}
  \label{theo:LowerBound4k+3}
  The worst case spanning ratio of the \graph{3} is at least \[\frac{3\cos\left(\frac{\theta}{4}\right)+\cos\left(\frac{3\theta}{4}\right)+\sin\left(\frac{\theta}{2}\right)+\sin \theta +\sin\left(\frac{3\theta}{2}\right)}{3\cos\left(\frac{\theta}{2}\right)+\cos\left(\frac{3\theta}{2}\right)}.\] 
\end{theorem}
\begin{proof}
We construct the lower bound example by extending the shortest path between two vertices $u$ and $w$ in three steps. We describe only how to extend one of the shortest paths between these vertices. To extend all shortest paths, the same modification is performed in each of the analogous cases, as shown in Figure~\ref{fig:LowerBound4k+3}. 

\begin{figure}[ht]
  \begin{center}
    \includegraphics{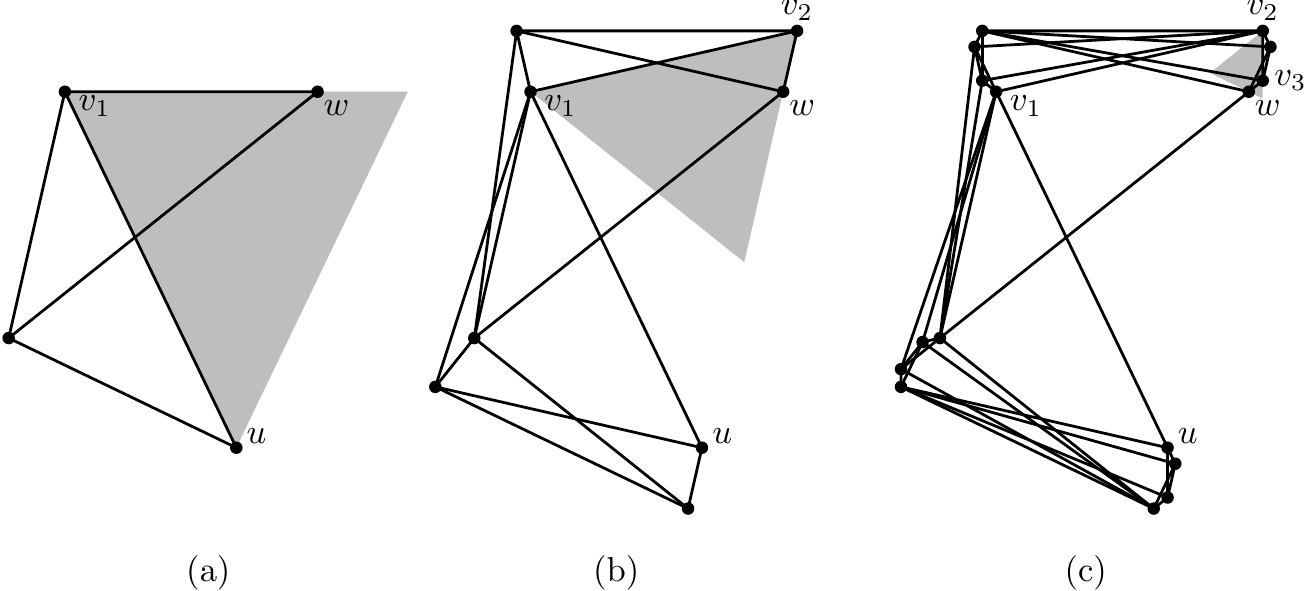}
  \end{center}
  \vspace{-1em}
  \caption{The construction of the lower bound for the \graph{3}}
  \label{fig:LowerBound4k+3}
\end{figure}

First, we place $w$ such that the angle between $u w$ and the bisector of the cone of $u$ that contains $w$ is $\theta/4$. Next, we ensure that there is no edge between $u$ and $w$ by placing a vertex $v_1$ in the upper corner of \canon{u}{w} that is furthest from $w$ (see Figure~\ref{fig:LowerBound4k+3}a). Next, we place a vertex $v_2$ in the corner of \canon{v_1}{w} that lies outside \canon{u}{w} (see Figure~\ref{fig:LowerBound4k+3}b). Finally, to ensure that there is no edge between $v_2$ and $w$, we place a vertex $v_3$ in \canon{v_2}{w} such that \canon{v_2}{w} and \canon{v_3}{w} have the same orientation (see Figure~\ref{fig:LowerBound4k+3}c). Note that we cannot place $v_3$ in the lower right corner of \canon{v_2}{w} since this would cause an edge between $u$ and $v_3$ to be added, creating a shortcut to $w$. 

One of the shortest paths in the resulting graph visits $u$, $v_1$, $v_2$, $v_3$, and $w$. Thus, to obtain a lower bound for the \graph{3}, we compute the length of this path. 

\begin{figure}[ht]
  \begin{center}
    \includegraphics{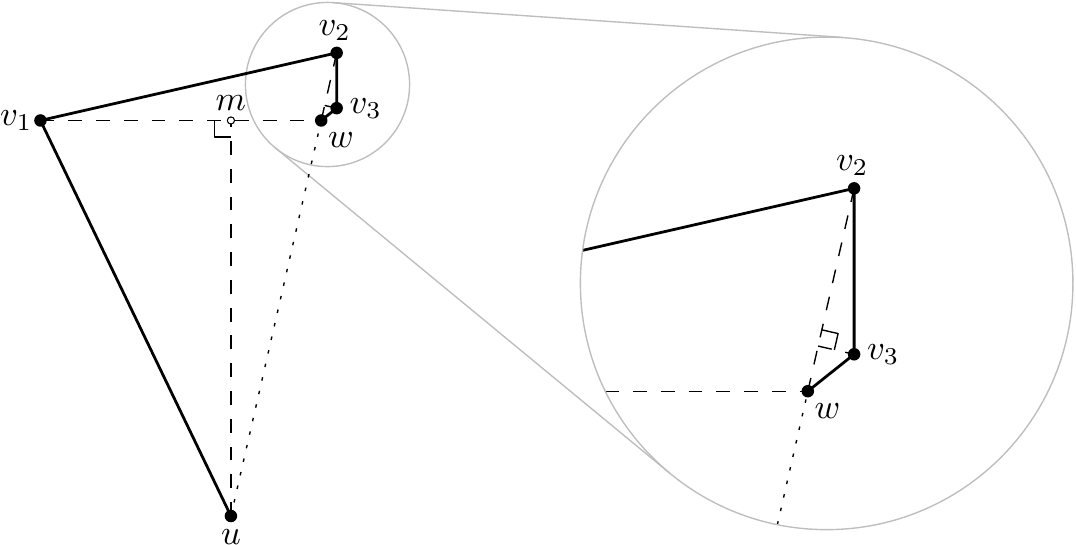}
  \end{center}
  \caption{The lower bound for the \graph{3}}
  \label{fig:LowerBoundComputation4k+3}
\end{figure}

Let $m$ be the midpoint of the side of \canon{u}{w} opposite $u$. By construction, we have that $\angle v_1 u m = \theta/2$, $\angle w u m = \angle v_2 v_1 w = \angle v_3 v_2 w = \theta/4$, $\angle v_3 w v_2 = 3\theta/4$, $\angle u v_1 w = \angle v_1 v_2 w = \pi/2 - \theta/2$, and $\angle v_2 v_3 w = \pi - \theta$ (see Figure~\ref{fig:LowerBoundComputation4k+3}). We can express the various line segments as follows: 
\begin{eqnarray*}
  |u v_1| &=& \frac{\cos \left( \frac{\theta}{4} \right)}{\cos \left( \frac{\theta}{2} \right)} \cdot |u w| \\ \\ 
  |v_1 w| &=& \frac{\sin \left( \frac{3\theta}{4} \right)}{\sin \left( \frac{\pi}{2} - \frac{\theta}{2} \right)} \cdot |u w| ~~=~~ \frac{\sin \left( \frac{3\theta}{4} \right)}{\cos \left( \frac{\theta}{2} \right)} \cdot |u w| \\ \\ 
  |v_1 v_2| &=& \frac{\cos \left( \frac{\theta}{4} \right)}{\cos \left( \frac{\theta}{2} \right)} \cdot |v_1 w| \\ \\
  |v_2 w| &=& \frac{\sin \left( \frac{\theta}{4} \right)}{\sin \left( \frac{\pi}{2} - \frac{\theta}{2} \right)} \cdot |v_1 w| ~~=~~ \frac{\sin \left( \frac{\theta}{4} \right)}{\cos \left( \frac{\theta}{2} \right)} \cdot |v_1 w| \\ \\ 
  |v_2 v_3| &=& \frac{\sin \left( \frac{3\theta}{4} \right)}{\sin (\pi - \theta)} \cdot |v_2 w| ~~=~~ \frac{\sin \left( \frac{3\theta}{4} \right)}{\sin (\theta)} \cdot |v_2 w| 
\end{eqnarray*}
\begin{eqnarray*}
  |v_3 w| &=& \frac{\sin \left( \frac{\theta}{4} \right)}{\sin (\pi - \theta)} \cdot |v_2 w| ~~=~~ \frac{\sin \left( \frac{\theta}{4} \right)}{\sin (\theta)} \cdot |v_2 w|
\end{eqnarray*}

Hence, the total length of the shortest path is $|u v_1| + |v_1 v_2| + |v_2 v_3| + |v_3 w|$, which can be rewritten to \[\frac{3\cos\left(\frac{\theta}{4}\right)+\cos\left(\frac{3\theta}{4}\right)+\sin\left(\frac{\theta}{2}\right)+\sin \theta +\sin\left(\frac{3\theta}{2}\right)}{3\cos\left(\frac{\theta}{2}\right)+\cos\left(\frac{3\theta}{2}\right)} \cdot |u w|,\] proving the theorem. 
\end{proof}

\subsection{Lower Bound on the \Graph{4}}
The \graph{2} has the nice property that any line perpendicular to the bisector of a cone is parallel to the boundary of a cone (Lemma~\ref{lem:Boundary}). As a result of this, if $u$, $v$, and $w$ are vertices with $v$ in one of the upper corners of \canon{u}{w}, then \canon{w}{v} is completely contained in \canon{u}{w}. The \graph{4} does not have this property. In this section, we show how to exploit this to construct a lower bound for the \graph{4} whose spanning ratio exceeds the worst case spanning ratio of the \graph{2}. 

\begin{theorem}
  The worst case spanning ratio of the \graph{4} is at least \[1 + 2 \tan \left( \frac{\theta}{2} \right) + 2 \tan^2 \left( \frac{\theta}{2} \right).\] 
\end{theorem}
\begin{proof}
We construct the lower bound example by extending the shortest path between two vertices $u$ and $w$ in three steps. We describe only how to extend one of the shortest paths between these vertices. To extend all shortest paths, the same modification is performed in each of the analogous cases, as shown in Figure~\ref{fig:LowerBound4k+4}. 

\begin{figure}[H]
  \begin{center}
    \includegraphics{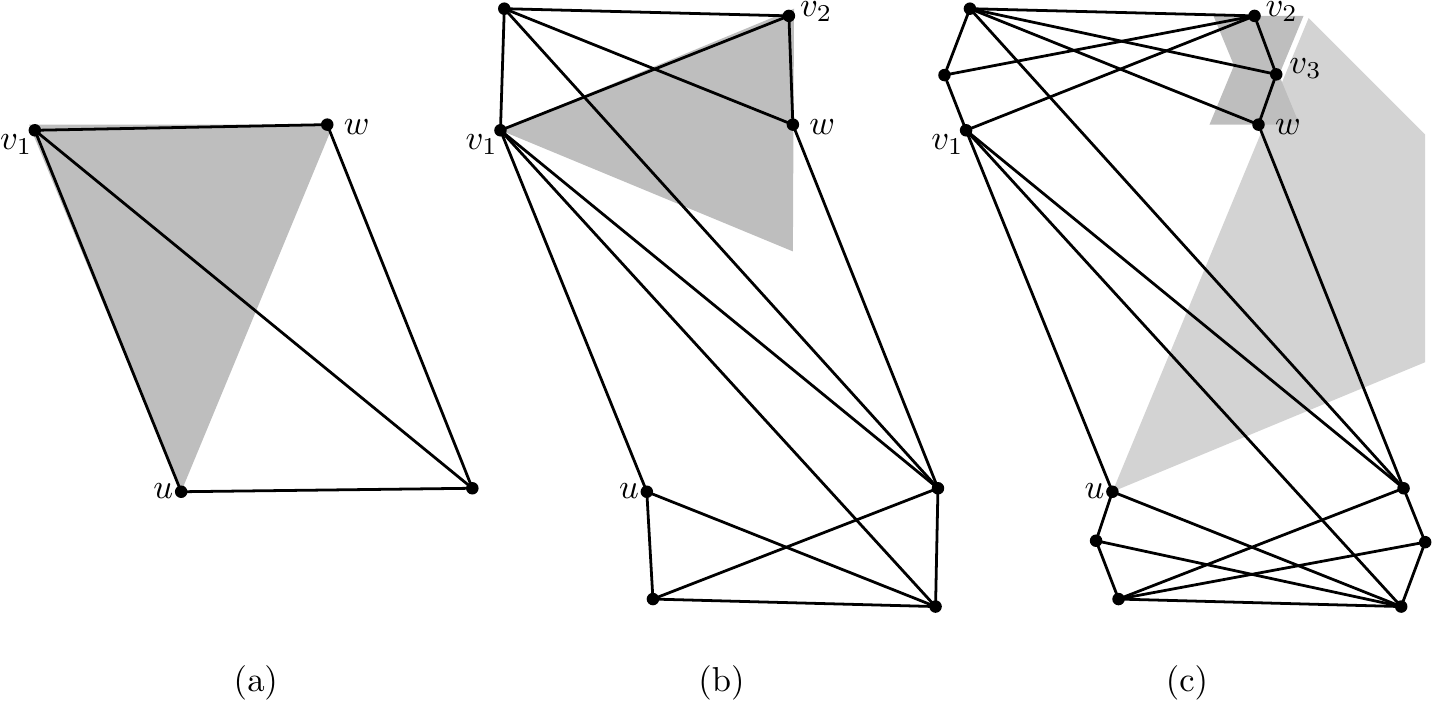}
  \end{center}
  \vspace{-1em}
  \caption{The construction of the lower bound for the \graph{4}}
  \label{fig:LowerBound4k+4}
\end{figure}

First, we place $w$ such that the angle between $u w$ and the bisector of the cone of $u$ that contains $w$ is $\theta/2$. Next, we ensure that there is no edge between $u$ and $w$ by placing a vertex $v_1$ in the upper corner of \canon{u}{w} that is furthest from $w$ (see Figure~\ref{fig:LowerBound4k+4}a). Next, we place a vertex $v_2$ in the corner of \canon{v_1}{w} that lies in the same cone of $u$ as $w$ and $v_1$ (see Figure~\ref{fig:LowerBound4k+4}b). Finally, we place a vertex $v_3$ in the intersection of the left boundary of \canon{v_2}{w} and the right boundary of \canon{w}{v_2} to ensure that there is no edge between $v_2$ and $w$ (see Figure~\ref{fig:LowerBound4k+4}c). Note that we cannot place $v_3$ in the lower right corner of \canon{v_2}{w} since this would cause an edge between $u$ and $v_3$ to be added, creating a shortcut to $w$. 

One of the shortest paths in the resulting graph visits $u$, $v_1$, $v_2$, $v_3$, and $w$. Thus, to obtain a lower bound for the \graph{4}, we compute the length of this path. 

\begin{figure}[H]
  \begin{center}
    \includegraphics{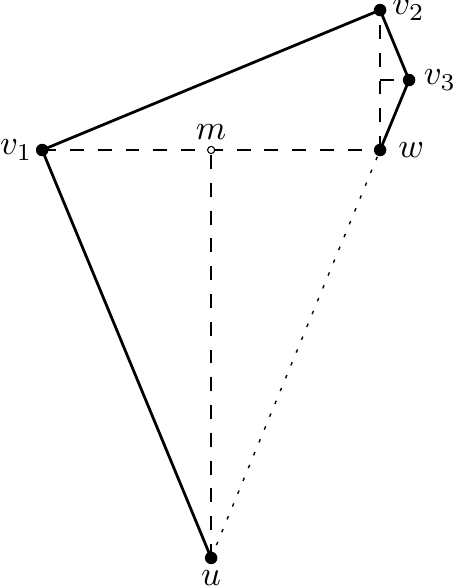}
  \end{center}
  \caption{The lower bound for the \graph{4}}
  \label{fig:LowerBoundComputation4k+4}
\end{figure}

Let $m$ be the midpoint of the side of \canon{u}{w} opposite~$u$. By construction, we have that $\angle v_1 u m = \angle w u m = \angle v_2 v_1 w = \angle v_3 v_2 w = \angle v_3 w v_2 = \theta/2$ (see Figure~\ref{fig:LowerBoundComputation4k+4}). We can express the various line segments as follows: 
\begin{eqnarray*}
  |u v_1| &=& |u w| \\
  |v_1 w| &=& 2 \sin \left( \frac{\theta}{2} \right) \cdot |u w| \\
  |v_1 v_2| &=& \frac{|v_1 w|}{\cos \left( \frac{\theta}{2} \right)} ~~=~~ 2 \tan \left( \frac{\theta}{2} \right) \cdot |u w| \\
  |v_2 w| &=& \tan \left( \frac{\theta}{2} \right) \cdot |v_1 w| ~~=~~ 2 \sin \left( \frac{\theta}{2} \right) \tan \left( \frac{\theta}{2} \right) \cdot |u w| \\
  |v_2 v_3| &=& |v_3 w| ~~=~~ \frac{\frac{1}{2}|v_1 w|}{\cos \left( \frac{\theta}{2} \right)} ~~=~~ \tan^2 \left( \frac{\theta}{2} \right) \cdot |u w| 
\end{eqnarray*}

Hence, the total length of the shortest path is $|u v_1| + |v_1 v_2| + |v_2 v_3| + |v_3 w|$, which can be rewritten to \[\left( 1 + 2 \tan \left( \frac{\theta}{2} \right) + 2 \tan^2 \left( \frac{\theta}{2} \right) \right) \cdot |u w|.\]  
\end{proof}

\subsection{Lower Bounds on the \Graph{5}}
In this section, we give a lower bound on the spanning ratio of the \graph{5}, for any integer $k \geq 1$. 

\begin{theorem}
  The worst case spanning ratio of the \graph{5} is at least \[ \frac{1}{2} \sqrt{4\sec\left(\frac{\theta}{2}\right) + 7\sec^2\left(\frac{\theta}{2}\right) + 4\sec^3\left(\frac{\theta}{2}\right) + \sec^4\left(\frac{\theta}{2}\right) - 8\cos\left(\frac{\theta}{2}\right) - 4}\] \[+ \tan\left(\frac{\theta}{2}\right) + \frac{1}{2}\sec\left(\frac{\theta}{2}\right)\tan\left(\frac{\theta}{2}\right).\]
\end{theorem}
\begin{proof}
We construct the lower bound example by extending the shortest path between two vertices $u$ and $w$ in two steps. We describe only how to extend one of the shortest paths between these vertices. To extend all shortest paths, the same modification is performed in each of the analogous cases, as shown in Figure~\ref{fig:LowerBound4k+5}. 

\begin{figure}[ht]
  \begin{center}
    \includegraphics{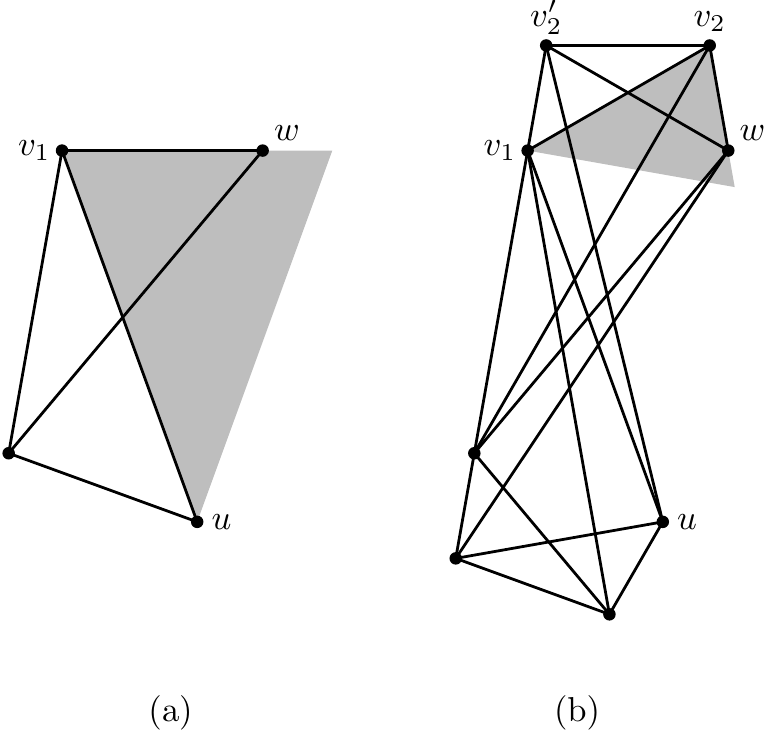}
  \end{center}
  \caption{The construction of the lower bound for the \graph{5}}
  \label{fig:LowerBound4k+5}
\end{figure}

First, we place $w$ such that the angle between $u w$ and the bisector of the cone of $u$ that contains $w$ is $\theta/4$. Next, we ensure that there is no edge between $u$ and $w$ by placing a vertex $v_1$ in the upper corner of \canon{u}{w} that is furthest from $w$ (see Figure~\ref{fig:LowerBound4k+5}a). Finally, we place a vertex $v_2$ in the corner of \canon{v_1}{w} that lies outside \canon{u}{w}. We also place a vertex $v_2'$ in the corner of \canon{w}{v_1} that lies in the same cone of $u$ as $w$ and $v_1$ (see Figure~\ref{fig:LowerBound4k+5}b). Note that placing $v_2'$ creates a shortcut between $u$ and $v_2'$, as $u$ is the closest vertex in one of the cones of $v_2'$. 

One of the shortest paths in the resulting graph visits $u$, $v_2'$, and $w$. Thus, to obtain a lower bound for the \graph{5}, we compute the length of this path. 

\begin{figure}[ht]
  \begin{center}
    \includegraphics{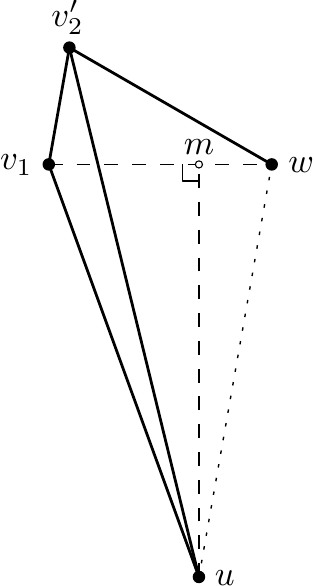}
  \end{center}
  \caption{The lower bound for the \graph{5}}
  \label{fig:LowerBoundComputation4k+5}
\end{figure}

Let $m$ be the midpoint of the side of \canon{u}{w} opposite $u$. By construction, we have that $\angle v_1 u m = \theta/2$, $\angle w u m = \theta/4$, $\angle v_1 w v_2' = 3\theta/4$, and $\angle u v_1 v_2' = \angle u v_1 w + \angle w v_1 v_2' = (\pi - \theta)/2 + (\pi - (\pi - \theta)/2 - 3\theta/4) = \pi - 3\theta/4$ (see Figure~\ref{fig:LowerBoundComputation4k+5}). We can express the various line segments as follows: 
\begin{eqnarray*}
  |u v_1| &=& \frac{\cos \left( \frac{\theta}{4} \right)}{\cos \left( \frac{\theta}{2} \right)} \cdot |u w| \\ \\
  |v_2' w| &=& \frac{\cos \left( \frac{\theta}{4} \right)}{\cos \left( \frac{\theta}{2} \right)} \cdot \left( \sin \left( \frac{\theta}{4} \right) + \cos \left( \frac{\theta}{4} \right) \tan \left( \frac{\theta}{2} \right) \right) \cdot |u w| \\ \\
  |v_1 v_2'| &=&  \left( \sin \left( \frac{\theta}{4} \right) + \cos \left( \frac{\theta}{4} \right) \tan \left( \frac{\theta}{2} \right) \right)^2 \cdot |u w| \\ \\
  |u v_2'| &=& \sqrt{|u v_1|^2 + |v_1 v_2'|^2 - 2 \cdot |u v_1| \cdot |v_1 v_2'| \cdot \cos \left( \pi - \frac{3\theta}{4} \right)}
\end{eqnarray*}

Hence, the total length of the shortest path is $|u v_2'| + |v_2' w|$, which can be rewritten to \[ \frac{1}{2} \sqrt{4\sec\left(\frac{\theta}{2}\right) + 7\sec^2\left(\frac{\theta}{2}\right) + 4\sec^3\left(\frac{\theta}{2}\right) + \sec^4\left(\frac{\theta}{2}\right) - 8\cos\left(\frac{\theta}{2}\right) - 4}\] \[+ \tan\left(\frac{\theta}{2}\right) + \frac{1}{2}\sec\left(\frac{\theta}{2}\right)\tan\left(\frac{\theta}{2}\right)\]  times the length of $u w$.
\end{proof}

\section{Comparison}
\label{sec:Comparison}
In this section we prove that the upper and lower bounds of the four families of $\theta$-graphs admit a partial ordering. We need the following lemma that can be proved by elementary calculus. 

\begin{lemma}
\label{lemma trigonometric inequalities}
  Let $x \in \left[0,\frac{\pi}{4}\right]$ be a real number. Then the following inequalities hold:
  \begin{enumerate}
    \item \label{lemma trigonometric inequalities item sin <}
    $\sin(x) \leq x$ with equality if and only if $x=0$.

    \item \label{lemma trigonometric inequalities item cos >}
    $\cos(x) \geq 1-\frac{x^2}{2}$ with equality if and only if $x=0$.

    \item \label{lemma trigonometric inequalities item sin >}
    $\sin(x) \geq x-\frac{x^3}{6}$ with equality if and only if $x=0$.

    \item \label{lemma trigonometric inequalities item cos <}
    $\cos(x) \leq 1-\frac{x^2}{2}+\frac{x^4}{24}$ with equality if and only if $x=0$.

    \item \label{lemma trigonometric inequalities item tan >}
    $\tan(x) \geq x$ with equality if and only if $x=0$.

    \item \label{lemma trigonometric inequalities item tan^2 >}
    $\tan^2\!(x) \geq x^2$ with equality if and only if $x=0$.
  \end{enumerate}
\end{lemma}

Using the above properties, we proceed to prove a number of relations between the four families of $\theta$-graphs. 

\setcounter{equation}{0}
\renewcommand{\theequation}{\alph{equation}}

\begin{lemma}
  \label{proposition inequalities}
  Let $ub(m)$ and $lb(m)$ denote the upper and lower bound on the $\theta_m$-graph: 
  \begin{align*}
    ub(m) &=
    \begin{cases}
    1 + 2 \sin\left(\frac{\pi}{4k+2}\right) & \textrm{if }m=4k+2\quad(k\geq 1)\cr
    &\cr
    \frac{\cos\left(\frac{\pi}{2(4k+3)}\right)}{\cos\left(\frac{\pi}{4k+3}\right) - \sin\left(\frac{3\pi}{2(4k+3)}\right)} & \textrm{if }m=4k+3\quad(k\geq 1)\cr
    &\cr
    1 + 2 \frac{\sin\left(\frac{\pi}{4k+4}\right)}{\cos\left(\frac{\pi}{4k+4}\right) - \sin\left(\frac{\pi}{4k+4}\right)} & \textrm{if }m=4k+4\quad(k\geq 1)\cr
    &\cr
    \frac{\cos\left(\frac{\pi}{2(4k+5)}\right)}{\cos\left(\frac{\pi}{4k+5}\right) - \sin\left(\frac{3\pi}{2(4k+5)}\right)} & \textrm{if }m=4k+5\quad(k\geq 1)\cr
    \end{cases}\\
    lb(m) &=
    \begin{cases}
    1 + 2 \sin\left(\frac{\pi}{4k+2}\right) & \textrm{\hspace{-1.5cm} if }m=4k+2\quad(k\geq 1)\cr
    &\cr
    \frac{3\cos\left(\frac{\pi}{2(4k+3)}\right)+\cos\left(\frac{3\pi}{2(4k+3)}\right)+\sin\left(\frac{\pi}{4k+3}\right)+\sin\left(\frac{2\pi}{4k+3}\right)+\sin\left(\frac{3\pi}{4k+3}\right)}{3\cos\left(\frac{\pi }{4k+3}\right)+\cos\left(\frac{3\pi }{4k+3}\right)} & \textrm{\hspace{-1.5cm} if }m=4k+3\quad(k\geq 1)\cr
    &\cr
    1+2\tan\left(\frac{\pi}{4k+4}\right)+2\tan^2\!\left(\frac{\pi}{4k+4}\right) & \textrm{\hspace{-1.5cm} if }m=4k+4\quad(k\geq 1)\cr
    &\cr
    \frac{\sqrt{4\sec\left(\frac{\pi}{4k+5}\right) + 7\sec^2\left(\frac{\pi}{4k+5}\right) + 4\sec^3\left(\frac{\pi}{4k+5}\right) + \sec^4\left(\frac{\pi}{4k+5}\right) - 8\cos\left(\frac{\pi}{4k+5}\right) - 4}}{2}  \cr +  \tan\left(\frac{\pi}{4k+5}\right) + \frac{1}{2}\sec\left(\frac{\pi}{4k+5}\right)\tan\left(\frac{\pi}{4k+5}\right) & \textrm{\hspace{-1.5cm} if }m=4k+5\quad(k\geq 1)\cr
    \end{cases}
  \end{align*}
  Then the following inequalities hold where $k$ is an integer.
  \begin{align}
    \label{4k+2 monotonic}
    ub(4(k+1)+2) &< lb(4k+2) \qquad (k \geq 1) \\
    \label{4k+3 monotonic}
    ub(4(k+1)+3) &< lb(4k+3) \qquad (k \geq 1) \\
    \label{4k+4 monotonic}
    ub(4(k+1)+4) &< lb(4k+4) \qquad (k \geq 1) \\
    \label{4k+5 monotonic}
    ub(4(k+1)+5) &< lb(4k+5) \qquad (k \geq 1) \\
    ub(4k+2) &< lb(4k+4) \qquad (k \geq 1) \\
    \label{4k+4 < 4(k-1)+2}
    ub(4(k+1)+4) &< lb(4k+2) \qquad (k \geq 1) \\
    ub(4(k+1)+5) &< lb(4k+3) \qquad (k \geq 1) \\
    ub(4(k+1)+3) &< lb(4k+5) \qquad (k \geq 1) \\
    ub(4k+5) &< lb(4k+2) \qquad (k \geq 2)
  \end{align}
\end{lemma}

\setcounter{equation}{2}
\renewcommand{\theequation}{\arabic{equation}}

\begin{proof}
We use the same strategy for each inequality. We use the definitions of $ub$ and $lb$ in combination with Lemma~\ref{lemma trigonometric inequalities}. Notice that the restriction on $k$ in each of these inequalities ensures that we can apply Lemma~\ref{lemma trigonometric inequalities}. We are then left with an algebraic inequality that can be translated into a polynomial inequality, which is easy to verify. 
\begin{enumerate}[(\alph{enumi})]
\Item 
\begin{align}
\nonumber & ~~~~ ub(4(k+1)+2) \\
\nonumber & = 1+2\sin\left(\frac{\pi}{4(k+1)+2}\right) & \textrm{by the definition of $ub$,} \\
\nonumber & < 1 + 2 \left(\frac{\pi}{4(k+1)+2}\right) & \textrm{by Lemma~\ref{lemma trigonometric inequalities}-\ref{lemma trigonometric inequalities item sin <},} \\
\label{proof eqn 1} & < 1 + 2\left(\left(\frac{\pi}{4k+2}\right)-\frac{1}{6}\left(\frac{\pi}{4k+2}\right)^3\right) & \textrm{see below,}\\
\nonumber & < 1+2\sin\left(\frac{\pi}{4k+2}\right) & \textrm{by Lemma~\ref{lemma trigonometric inequalities}-\ref{lemma trigonometric inequalities item sin >},} \\
\nonumber & = lb(4k+2) & \textrm{by the definition of $lb$.}
\end{align}
We now explain why~\eqref{proof eqn 1} holds. The inequality \[1+2\left(\frac{\pi}{4(k+1)+2}\right) < 1 + 2\left(\left(\frac{\pi}{4k+2}\right)-\frac{1}{6}\left(\frac{\pi}{4k+2}\right)^3\right)\] can be simplified to
\begin{align}
\label{ineq polyn}
192k^2 + \left(192-2\pi^2\right) k +\left(48-3\pi^2\right) > 0.
\end{align}
The largest real root of the polynomial involved in~\eqref{ineq polyn} is negative. Moreover, \eqref{proof eqn 1} holds for $k=1$. Therefore, \eqref{proof eqn 1} holds for any $k\geq 1$.

\item The proof is analogous to the one of~\eqref{4k+2 monotonic}.

\item The proof is analogous to the one of~\eqref{4k+2 monotonic}.

\item We let
\begin{align*}
f(k) =~ & \frac{\cos\left(\frac{\pi}{2(4(k+1)+5)}\right)}{\cos\left(\frac{\pi}{4(k+1)+5}\right)-\sin\left(\frac{3\pi}{2(4(k+1)+5)}\right)} ,\\
r(k) =~ & 4\sec\left(\frac{\pi}{4k+5}\right) + 7\sec^2\left(\frac{\pi}{4k+5}\right) + 4\sec^3\left(\frac{\pi}{4k+5}\right) + \\ 
& \sec^4\left(\frac{\pi}{4k+5}\right) - 8\cos\left(\frac{\pi}{4k+5}\right) - 4,\\
g(k) =~ & 2\tan\left(\frac{\pi}{4k+5}\right)+\sec\left(\frac{\pi}{4k+5}\right)\tan\left(\frac{\pi}{4k+5}\right),
\end{align*}
so that
\begin{align*}
ub(4(k+1)+5) &= f(k), \\
lb(4k+5) &= \frac{\sqrt{r(k)}+g(k)}{2}.
\end{align*}
Using a proof similar to the one of~\eqref{4k+2 monotonic}, we can prove that
\begin{align*}
(2f(k)-g(k))^2 &< r(k). 
\end{align*}
Using a proof similar to the one of~\eqref{4k+2 monotonic}, we can prove that $2f(k)-g(k) > 0$, for $k \geq 1$, thus we can proceed as follows
\begin{align*}
2f(k)-g(k) &< \sqrt{r(k)} \\
f(k) &< \frac{\sqrt{r(k)}+g(k)}{2} \\
ub(4(k+1)+5) &< lb(4k+5), 
\end{align*}
for $k \geq 1$.

\item The proof is analogous to the one of~\eqref{4k+2 monotonic}.

\item The proof is analogous to the one of~\eqref{4k+2 monotonic}.

\item The proof is analogous to the one of~\eqref{4k+5 monotonic}.

\item The proof is analogous to the one of~\eqref{4k+5 monotonic}.

\item The proof is analogous to the one of~\eqref{4k+5 monotonic}.
\end{enumerate}
\vspace{-1.5em}
\end{proof}

We note that inequalities~\eqref{4k+2 monotonic}, \eqref{4k+3 monotonic}, \eqref{4k+4 monotonic}, and~\eqref{4k+5 monotonic} imply that the spanning ratio is monotonic within each of the four families. We also note that increasing the number of cones of a $\theta$-graph by 2 from $4k + 2$ to $4k + 4$ increases the worst case spanning ratio, thus showing that adding cones can make the spanning ratio worse instead of better. Therefore, the spanning ratio is non-monotonic between families. 

\begin{corollary}
  We have the following partial order on the spanning ratios of the four families (see Figure~\ref{figure partial order}).
\end{corollary}

\begin{figure}[ht]
  \centering
  \includegraphics[scale=1]{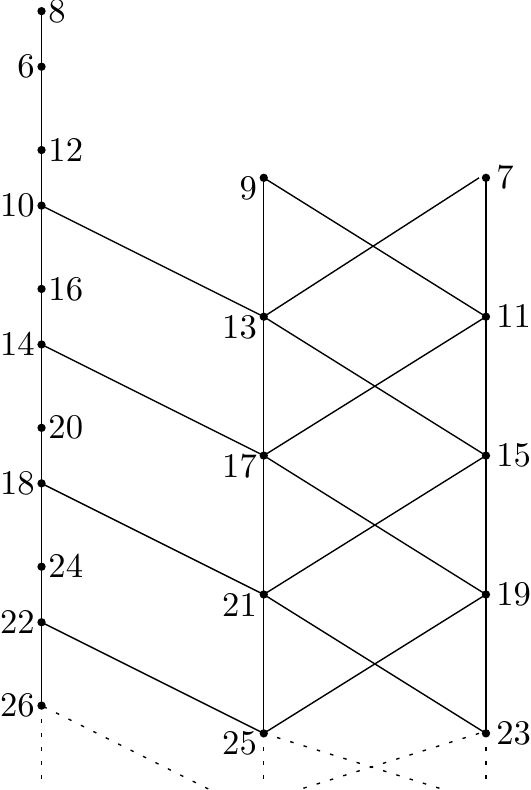}
  \caption{Partial order on the spanning ratios of the four families}
  \label{figure partial order}
\end{figure}

\section{Tight Routing Bounds}
While improving the upper bounds on the spanning ratio of the \graph{4}, we also improved the upper bound on the routing ratio of the $\theta$-routing algorithm. In this section we show that this bound of $1 + 2 \sin (\theta/2) / \left( \cos (\theta/2) - \sin (\theta/2) \right)$ and the current upper bound of $1 / \left( 1 - 2 \sin (\theta/2) \right)$ on the $\theta_{10}$-graph are tight, i.e. we provide matching lower bounds on the routing ratio of the $\theta$-routing algorithm on these families of graphs.

\subsection{Tight Routing Bounds for the \Graph{4}}
In this section we show that the upper bound of $1 + (2 \sin (\theta/2)) / (\cos (\theta/2) - \sin (\theta/2))$ on the routing ratio of the $\theta$-routing algorithm for the \graph{4} is a tight bound. 

\begin{theorem}
  The $\theta$-routing algorithm is $\left( 1 + \frac{2 \sin \left( \frac{\theta}{2} \right)}{\cos \left( \frac{\theta}{2} \right) - \sin \left( \frac{\theta}{2} \right)} \right)$-competitive on the \graph{4} and this bound is tight. 
\end{theorem}
\begin{proof}
  Corollary~\ref{cor:Routing4k+4} showed that the routing ratio is at most $1 + (2 \sin (\theta/2)) / (\cos (\theta/2) - \sin (\theta/2))$, hence it suffices to show that this is also a lower bound. 

  We construct the lower bound example on the competitiveness of the $\theta$-routing algorithm on the \graph{4} by repeatedly extending the routing path from source $u$ to destination $w$. First, we place $w$ in the right corner of \canon{u}{w}. To ensure that the $\theta$-routing algorithm does not follow the edge between $u$ and $w$, we place a vertex $v_1$ in the left corner of \canon{u}{w}. Next, to ensure that the $\theta$-routing algorithm does not follow the edge between $v_1$ and $w$, we place a vertex $v_1'$ in the left corner of \canon{v_1}{w}. We repeat this step until we have created a cycle around $w$ (see Figure~\ref{fig:Routing4k+4TightSteps}a). 

  \begin{figure}[ht]
    \centering
    \includegraphics{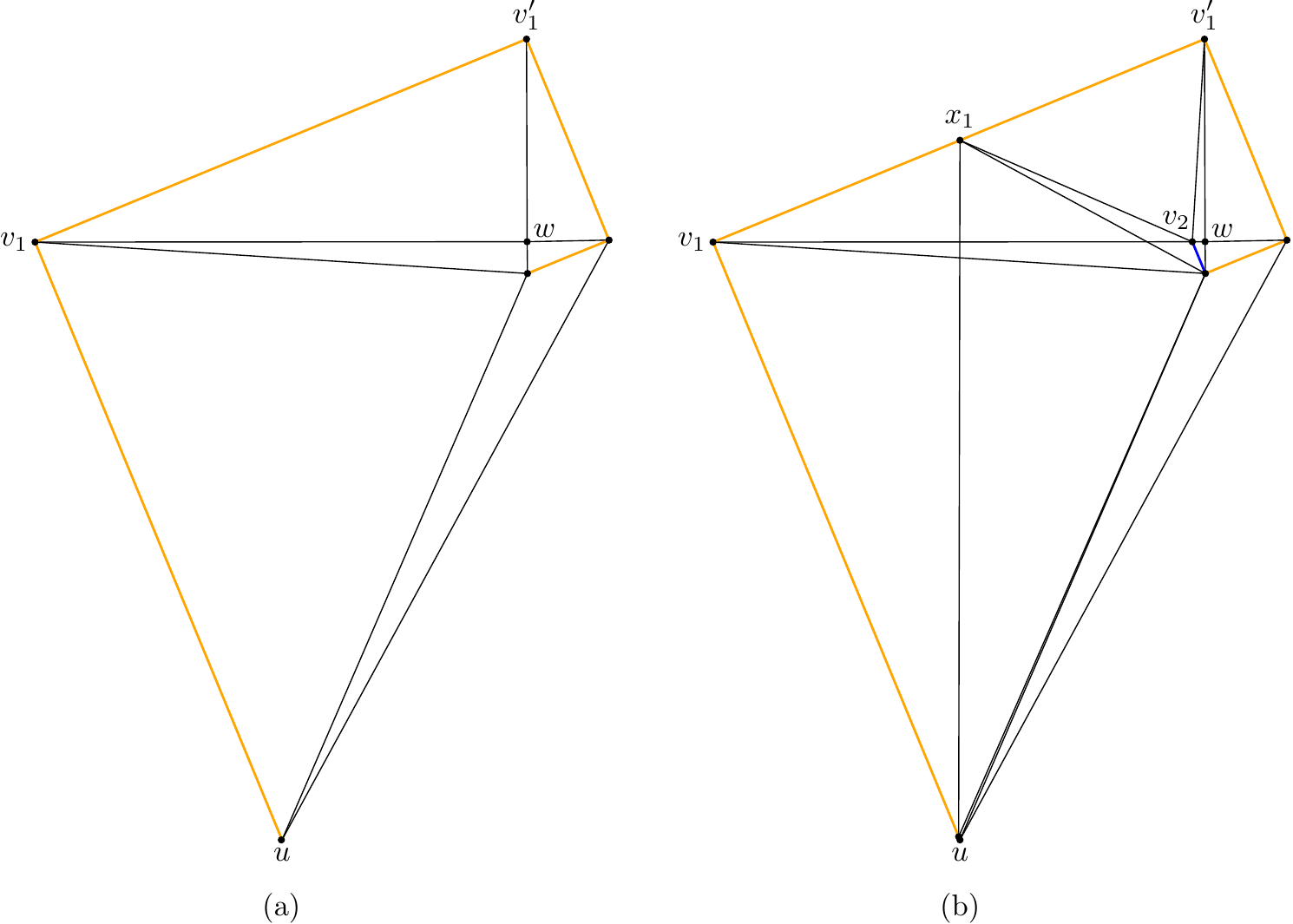}
    \caption{Constructing a lower bound example for $\theta$-routing on the \graph{4}: (a) after constructing the first cycle, (b) after adding $v_2$, the first vertex of the second cycle, and $x_1$, the auxiliary vertex needed to maintain the first cycle}
    \label{fig:Routing4k+4TightSteps}
  \end{figure}

  To extend the routing path further, we again place a vertex $v_2$ in the corner of the current canonical triangle. To ensure that the routing algorithm still routes to $v_1$ from $u$, we place $v_2$ slightly outside of \canon{u}{v_1}. However, another problem arises: vertex $v_1'$ is no longer the vertex closest to $v_1$ in \canon{v_1}{w}, as $v_2$ is closer. To solve this problem, we also place a vertex $x_1$ in \canon{v_1}{v_2} such that $v_1'$ lies in \canon{x_1}{w} (see Figure~\ref{fig:Routing4k+4TightSteps}b). By repeating this process four times, we create a second cycle around $w$. 

  To add more cycles around $w$, we repeat the same process as described above: place a vertex in the corner of the current canonical triangle and place an auxiliary vertex to ensure that the previous cycle stays intact. Note that when placing $x_i$, we also need to ensure that it does not lie in \canon{x_{i-1}}{w}, to prevent shortcuts from being formed. A lower bound example consisting of two cycles is shown in Figure~\ref{fig:Routing4k+4Tight}.

  \begin{figure}[ht]
    \centering
    \includegraphics{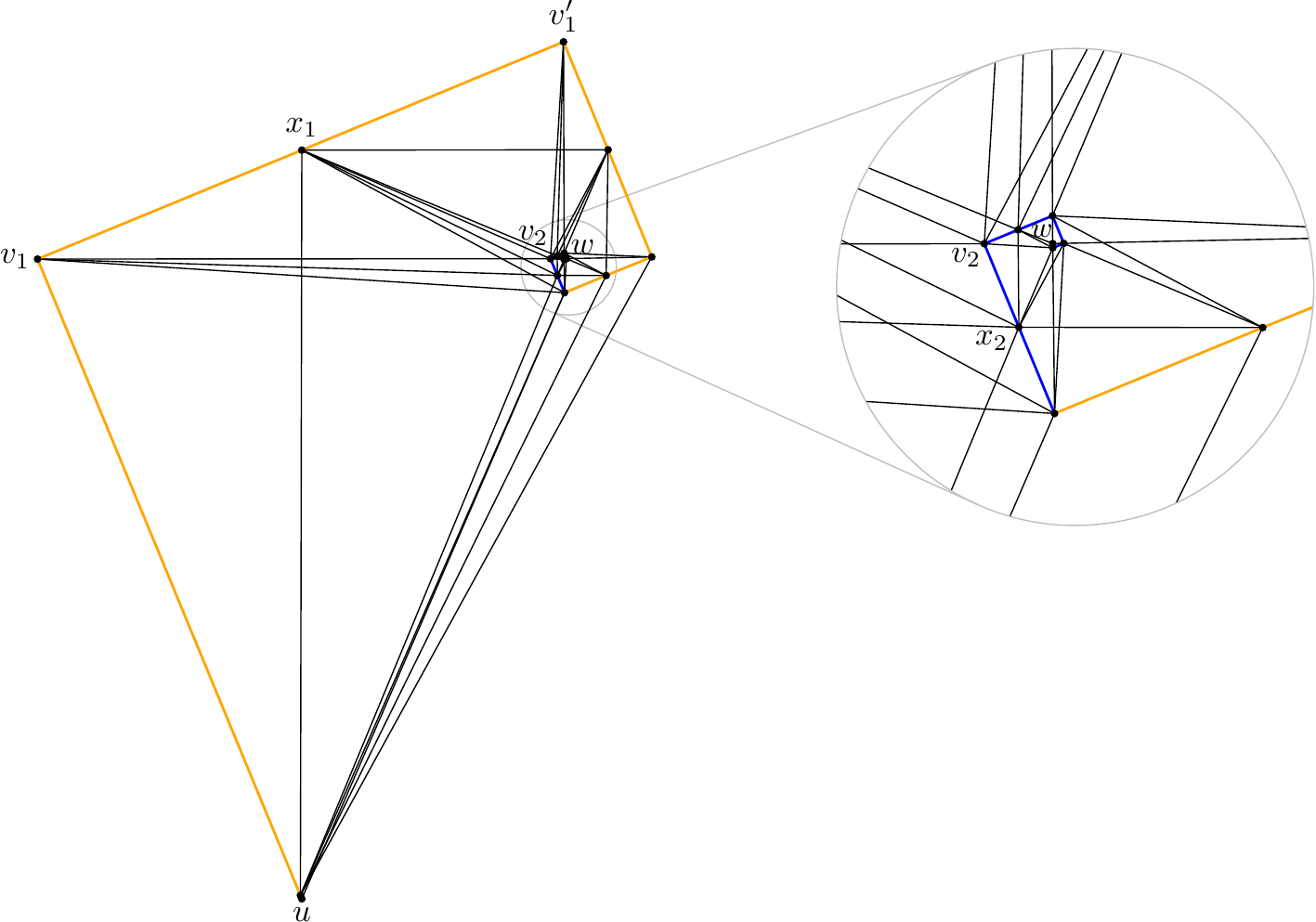}
    \caption{A lower bound example for $\theta$-routing on the \graph{4}, consisting of two cycles: the first cycle is coloured orange and the second cycle is coloured blue}
    \label{fig:Routing4k+4Tight}
  \end{figure}

  This way we need to add auxiliary vertices only to the $(k-1)$-th cycle, when adding the $k$-th cycle, hence we can add an additional cycle using only a constant number of vertices. Since we can place the vertices arbitrarily close to the corners of the canonical triangles, we ensure that $|u v_1| = |u w|$ and that the distance between consecutive vertices $v_i$ and $v_i'$ is always $1 / \cos (\theta/2)$ times $|v_i w|$. Hence, when we take $|u w| = 1$ and let the number of vertices approach infinity, we get that the total length of the path is $1 + 2 \sin (\theta/2) \cdot \sum_{i=0}^\infty (\tan^i (\theta/2) / \cos (\theta/2))$, which can be rewritten to $1 + (2 \sin (\theta/2)) / (\cos (\theta/2) - \sin (\theta/2))$.

\end{proof}

\subsection{Tight Routing Bounds for the $\theta_{10}$-Graph}
In this section we show that the upper bound of $1/(1 - 2 \sin (\theta/2))$ on the routing ratio of the $\theta$-routing algorithm for the $\theta_{10}$-graph is a tight bound. 

\begin{theorem}
  The $\theta$-routing algorithm is $\left( 1 / \left(1 - 2 \sin \left( \theta/2 \right) \right) \right)$-competitive on the $\theta_{10}$-graph and this bound is tight. 
\end{theorem}
\begin{proof}
  Ruppert and Seidel~\cite{RS91} showed that the routing ratio is at most $1/(1 - 2 \sin (\theta/2))$, hence it suffices to show that this is also a lower bound. 

  We construct the lower bound example on the competitiveness of the $\theta$-routing algorithm on the $\theta_{10}$-graph by repeatedly extending the routing path from source $u$ to destination $w$. First, we place $w$ in the right corner of \canon{u}{w}. To ensure that the $\theta$-routing algorithm does not follow the edge between $u$ and $w$, we place a vertex $v_1$ in the left corner of \canon{u}{w}. Next, to ensure that the $\theta$-routing algorithm does not follow the edge between $v_1$ and $w$, we place a vertex $v_1'$ in the left corner of \canon{v_1}{w}. We repeat this step until we have created a cycle around $w$ (see Figure~\ref{fig:Routing10Tight}). 

  \begin{figure}[ht]
    \centering
    \includegraphics{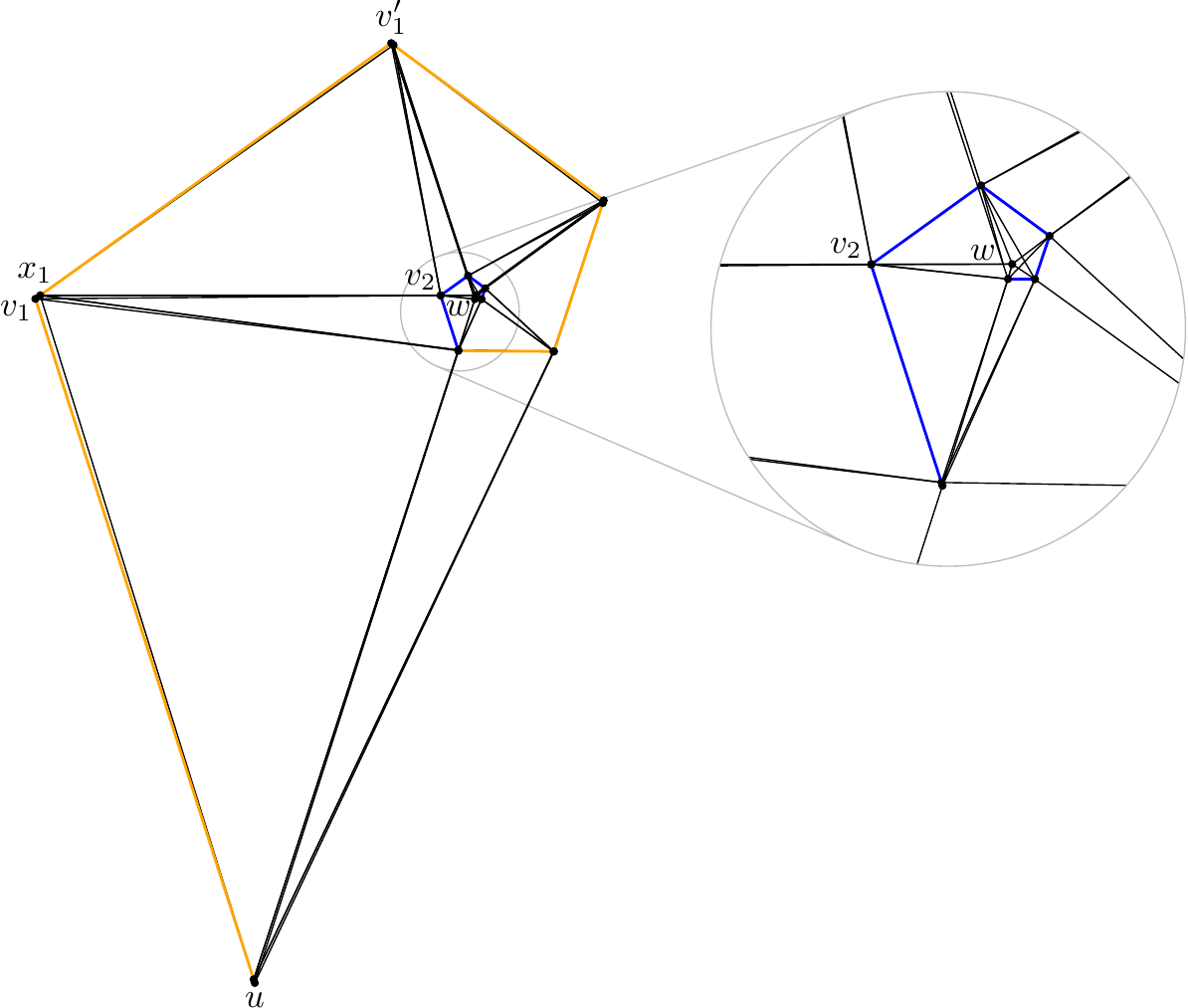}
    \caption{A lower bound example for $\theta$-routing on the $\theta_{10}$-graph, consisting of two cycles: the first cycle is coloured orange and the second cycle is coloured blue}
    \label{fig:Routing10Tight}
  \end{figure}

  To extend the routing path further, we again place a vertex $v_2$ in the corner of the current canonical triangle. To ensure that the routing algorithm still routes to $v_1$ from $u$, we place $v_2$ slightly outside of \canon{u}{v_1}. However, another problem arises: vertex $v_1'$ is no longer the vertex closest to $v_1$ in \canon{v_1}{w}, as $v_2$ is closer. To solve this problem, we also place a vertex $x_1$ in \canon{v_1}{v_2} such that $v_1'$ lies in \canon{x_1}{w} (see Figure~\ref{fig:Routing10Zoomed}). By repeating this process four times, we create a second cycle around $w$. 

  \begin{figure}[ht]
    \centering
    \includegraphics{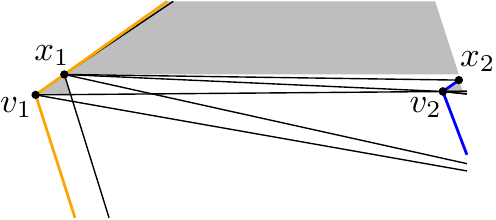}
    \caption{The placement of vertices such that previous cycles stay intact when adding a new cycle}
    \label{fig:Routing10Zoomed}
  \end{figure}

  To add more cycles around $w$, we repeat the same process as described above: place a vertex in the corner of the current canonical triangle and place an auxiliary vertex to ensure that the previous cycle stays intact. Note that when placing $x_i$, we also need to ensure that it does not lie in \canon{x_{i-1}}{w}, to prevent shortcuts from being formed (see Figure~\ref{fig:Routing10Zoomed}). This means that in general $x_i$ does not lie arbitrarily close to the corner of \canon{v_i}{v_{i+1}}. 

  This way we need to add auxiliary vertices only to the $(k-1)$-th cycle, when adding the $k$-th cycle, hence we can add an additional cycle using only a constant number of vertices. Since we can place the vertices arbitrarily close to the corners of the canonical triangles, we ensure that the distance to $w$ is always $2 \sin (\theta/2)$ times the distance between $w$ and the previous vertex along the path. Hence, when we take $|u w| = 1$ and let the number of vertices approach infinity, we get that the total length of the path is $\sum_{i=0}^\infty \left(2 \sin (\theta/2) \right)^i$, which can be rewritten to $1 / \left(1 - 2 \sin (\theta/2) \right)$.
\end{proof}

\section{Conclusion}
We showed that the \graph{2} has a tight spanning ratio of $1 + 2 \sin(\theta/2)$. This is the first time tight spanning ratios have been found for a large family of $\theta$-graphs. Previously, the only $\theta$-graph for which tight bounds were known was the $\theta_6$-graph. We also gave improved upper bounds on the spanning ratio of the \graph{3}, the \graph{4}, and the \graph{5}. 

We also constructed lower bounds for all four families of $\theta$-graphs and provided a partial order on these families. In particular, we showed that the \graph{4} has a spanning ratio of at least $1 + 2 \tan(\theta/2) + 2 \tan^2(\theta/2)$. This result is somewhat surprising since, for equal values of $k$, the worst case spanning ratio of the \graph{4} is greater than that of the \graph{2}, showing that increasing the number of cones can make the spanning ratio worse. 

There remain a number of open problems, such as finding tight spanning ratios for the \graph{3}, the \graph{4}, and the \graph{5}. Similarly, for the $\theta_4$ and $\theta_5$-graphs, though upper and lower bounds are known, these are far from tight. It would also be nice if we could improve the routing algorithms for $\theta$-graphs. At the moment, $\theta$-routing is the standard routing algorithm for general $\theta$-graphs, but it is unclear whether this is the best routing algorithm for general $\theta$-graphs: though we showed that the current bounds on the competitiveness of the $\theta$-routing algorithm are tight in case of the \graph{4}, this does not imply that there exists no algorithm that can do better on these graphs. As a special case, we note that the $\theta$-routing algorithm is not $o(n)$-competitive on the $\theta_6$-graph, but a better (tight) algorithm is known to exist~\cite{BFRV12}.

\bibliographystyle{abbrv}
\bibliography{references}

\begin{thebibliography}{1}

\bibitem{BGHI10}
N.~Bonichon, C.~Gavoille, N.~Hanusse, and D.~Ilcinkas.
\newblock Connections between theta-graphs, {D}elaunay triangulations, and
  orthogonal surfaces.
\newblock In {\em Proceedings of the 36th International Conference on Graph
  Theoretic Concepts in Computer Science (WG 2010)}, pages 266--278, 2010.

\bibitem{BFRV12}
P.~Bose, R.~Fagerberg, A.~van Renssen, and S.~Verdonschot.
\newblock Competitive routing in the half-$\theta_6$-graph.
\newblock In {\em Proceedings of the 23rd ACM-SIAM Symposium on Discrete
  Algorithms (SODA 2012)}, pages 1319--1328, 2012.

\bibitem{BS11}
P.~Bose and M.~Smid.
\newblock On plane geometric spanners: A survey and open problems.
\newblock {\em Computational Geometry: Theory and Applications (CGTA)},
  46(7):818--830, 2013.

\bibitem{Chew89}
P.~Chew.
\newblock There are planar graphs almost as good as the complete graph.
\newblock {\em Journal of Computer and System Sciences (JCSS)}, 39(2):205--219,
  1989.

\bibitem{Cl87}
K.~Clarkson.
\newblock Approximation algorithms for shortest path motion planning.
\newblock In {\em Proceedings of the 19th Annual ACM Symposium on Theory of
  Computing (STOC 1987)}, pages 56--65, 1987.

\bibitem{Keil88}
J.~Keil.
\newblock Approximating the complete {E}uclidean graph.
\newblock In {\em Proceedings of the 1st Scandinavian Workshop on Algorithm
  Theory (SWAT 1988)}, pages 208--213, 1988.

\bibitem{NS-GSN-06}
G.~Narasimhan and M.~Smid.
\newblock {\em Geometric Spanner Networks}.
\newblock Cambridge University Press, 2007.

\bibitem{RS91}
J.~Ruppert and R.~Seidel.
\newblock Approximating the $d$-dimensional complete {E}uclidean graph.
\newblock In {\em Proceedings of the 3rd Canadian Conference on Computational
  Geometry (CCCG 1991)}, pages 207--210, 1991.

\end{thebibliography}

\end{document}